\documentclass[10pt,twocolumn,twoside]{IEEEtran}

\IEEEoverridecommandlockouts                              

\usepackage{graphics} 
\usepackage{epsfig} 
\usepackage{mathptmx} 
\usepackage{times} 
\usepackage{amsmath} 
\usepackage{amssymb}  
\usepackage{dsfont}
\usepackage{subfigure}
\usepackage{cancel}
\usepackage{margins}
\usepackage{mathtools}
\usepackage{multirow}  
\usepackage{setspace}
\newtheorem{thm}{Theorem}[section]
\newtheorem{defn}{Definition}[section]
\newtheorem{cor}{Corollary}[section]
\newtheorem{lem}{Lemma}[section]
\newtheorem{remark}{Remark}[section]

\title{Synchronization of Multi-Agent Systems With Heterogeneous Controllers}

\author{Anoop Jain and Debasish Ghose 
\thanks{A. Jain is a graduate student at the Guidance, Control and Decision System Laboratory (GCDSL) in the Department of
                 Aerospace Engineering, Indian Institute of Science,
                  Bangalore, India (email: anoopj@aero.iisc.ernet.in; anoopjain.jn@gmail.com).}
\thanks{D. Ghose is a Professor at the Guidance, Control and Decision System Laboratory (GCDSL) in the Department of
              Aerospace Engineering, Indian Institute of
              Science, Bangalore, India (email: dghose@aero.iisc.ernet.in).}}

\begin{document}
\maketitle


\begin{abstract}
This paper studies the synchronization of a multi-agent system where the agents are coupled through heterogeneous controller gains. Synchronization refers to the situation where all the agents in a group have a common velocity direction. We generalize existing results and show that by using heterogeneous controller gains, the final velocity direction at which the system of agents synchronize can be controlled. The effect of heterogeneous gains on the reachable set of this final velocity direction is further analyzed. We also show that for realistic systems, a limited control force to stabilize the agents to the synchronized condition can be achieved by confining these heterogeneous controller gains to an upper bound. Simulations are given to support the theoretical findings.
\end{abstract}



\section{Introduction}
The phenomenon of collective synchronization in a network of coupled oscillators has attracted the attention of many researchers from different disciplines of science and engineering. Particularly, in the field of control engineering, rendezvous, consensus, and formation of multi-agent systems \cite{Beard2008}, \cite{Mesbhai2010}, are some variants of this fascinating phenomenon. In this paper, we study synchronization of a multi-agent system, which is achieved when, at all times, the agents and their position centroid, have a common velocity direction. Complementary to synchronization is the phenomenon of balancing, which refers to the situation in which agents move in such a way that their position centroid remains fixed. In this paper, only the problem of synchronization is considered.

Recently, the important insights in understanding the phenomenon of synchronization have come from the study of the Kuramoto model \cite{Strogatz2000}. This model is widely studied in the literature in the context of achieving synchronization and balancing in multi-agent systems. For instance in \cite{Sepulchre2007}, Kuramoto model type steering control law is derived to stabilize synchronized and balanced formations in a group of agents. The proposed control law in \cite{Sepulchre2007} operates with homogeneous controller gains, which gives rise to average consensus in the initial heading angles of the agents. Recently, the effect of heterogeneity in various aspects have been studied in the literature. For example, \cite{Seyboth2014} considers heterogeneous velocities of the agents. In a similar spirit, in this paper, we consider that the controller gains are heterogeneously distributed, that is, they are not necessarily the same for each agent, and can be deterministically varied. It will be shown that this type of heterogeneity in the controller gains also leads to a synchronized formation, in which a desired final velocity direction can be obtained by a proper selection of gains. Some preliminary results on this problem have been earlier obtained in \cite{Jain2013}.

The motivation to study synchronization under heterogeneous controller gains is twofold. First, in many engineering applications related to aerial and underwater vehicles, it is required that the vehicles move in a formation. For example, formations of Unmanned Aerial Vehicles (UAVs), Autonomous Underwater Vehicles (AUVs), and wheeled mobile robots are widely used for tracking, search, surveillance, oceanic explorations, etc. In addition to the basic requirement of maintaining formation in terms of connectivity preservation \cite{hasan2015}, one can utilize heterogeneity in the controller gains so that the formation of these vehicles can be made to move in a desired direction, thus helping to explore an area of interest more effectively. Secondly, while implementing the control law physically for the homogeneous gains case, it is impossible to get identical controller gain for each agent. Thus, some error in the individual controller gains is inevitable, leading to heterogeneity in the controller gains. It would be useful to know the effect of this heterogeneity on the synchronization performance of the multi-agent system.

Synchronization and its various aspects are widely studied in the literature. In \cite{Dong2015}, finite-time phase-frequency synchronization of Kuramoto oscillators is discussed. To achieve this finite-time convergence, the Kuramoto model is modified as a normalized and signed gradient system. A generalization of Kuramoto model in which the oscillators are coupled by both positive and negative coupling strength is given in \cite{Hong2011}. It is shown that the oscillators with positive coupling are attracted to the mean field and tend to synchronize with it. While oscillators with negative coupling are repelled by the mean field and prefer a phase diametrically opposed to it. In order to achieve complete phase and frequency synchronization of the Kuramoto model, Jadbabaie et al. in \cite{Jadbabaie2004} show that there is a critical value of the coupling below which a totally synchronized state does not exist. Chopra and Spong in \cite{Chopra2005} and \cite{Chopra2009} provide an improved bound on the coupling parameter to ensure exponential synchronization of the natural frequencies of all oscillators to the mean natural frequency of the group. Similar results are given in Ha et al. \cite{Ha2010} who provide sufficient conditions for the initial configurations of oscillators toward their complete synchronization. In \cite{Dorfler2011}, various bounds on the critical coupling strength for synchronization in the Kuramoto model, are presented. The problems of cooperative uniform and exponential synchronization in multi-agent systems are discussed in \cite{Monshizadeh2015} and \cite{Ma2015}, respectively. A survey on synchronization in complex networks of phase oscillators is presented in \cite{Dorfler2014}. In \cite{Sinha2006}, heterogeneous controller gains have been used in a cyclic pursuit framework to obtain desired meeting points (rendezvous) and directions. Moreover, the idea of dynamically adjustable control gains have been used in \cite{Ding2012} to study the pursuit formation of multiple autonomous agents.

The content of the rest of this paper is organized as follows. In Section II, we describe the dynamics of the system and formulate the problem. In Section III, we analyze the effect of heterogeneous controller gains on the final velocity direction at which the system of agents synchronize. In Section IV, by deriving a less restrictive condition on the heterogeneous gains for a special case of two agents, we show that the reachable set of the final velocity direction further expands. In order to model realistic systems,
a bound on the control force is obtained in Section V. Finally, we conclude the paper by summarizing the main results in Section VI.

\section{System Description and Problem Formulation}
\subsection{System Model}
Consider a multi-agent system composed of $N$ agents in which each agent, assumed to have unit mass, moves at unit speed in a $2-$dimensional plane. We identify this $2-$dimensional plane $\mathbb{R}^2$ with the complex plane $\mathbb{C}$ and use complex variables to describe the position and velocity of each agent. For $k = 1, \ldots, N$, the position of the $k^\text{th}$ agent is $r_k = x_k + iy_k\in \mathbb{C}$, while the velocity of the $k^\text{th}$ agent is $\dot{r}_k = e^{i\theta_k} = \cos\theta_k + i\sin\theta_k \in \mathbb{C}$, where, $\theta_k$ is the orientation of the (unit) velocity vector of the $k^\text{th}$agent from the real axis, and $i = \sqrt{-1}$ denotes the standard complex number. The orientation, $\theta_k$ of the velocity vector, which is also referred to as the phase of the $k^\text{th}$ agent \cite{Strogatz2000}, represents a point on the unit circle $S^1$. With these notations, the equations of motion of the $k^\text{th}$ agent are
\begin{subequations}\label{modelNew}
\begin{align}
\label{modelNew1}\dot{r}_k & = e^{i\theta_k}\\
\label{modelNew2} \dot{\theta}_k & =  u_k; ~~~ k = 1, \ldots, N,
\end{align}
\end{subequations}
where, $u_k \in \mathbb{R}$ is the feedback control law, which controls the angular rate of the $k^\text{th}$ agent. If, $\forall k$, the control law $u_k$ is identically zero, then each agent travels at constant unit speed in a straight line in its initial direction $\theta_k(0)$ and its motion is decoupled from the other agents. If, $\forall k$, the control input $u_k = \omega_0$ is constant and non zero then each agent rotates on a circle of radius $|\omega_0|^{-1}$. The convention of direction of rotation on the circle followed in this paper is, if $\omega_0 > 0~(\omega_0 < 0)$, then all the agents rotate in the anticlockwise (clockwise) direction.

Note that the agent's dynamics given by \eqref{modelNew} describe a unicycle model, and is widely studied in the literature \cite{Lin2005}$-$\cite{Egerstedt2001}. Furthermore, the control algorithms proposed in this paper are decentralized, and there is no centralized information available to the agents which lead them to synchronize at a desired velocity direction. Only the heterogeneity in the controller gains is a mean to steer the agents towards synchronization at a desired velocity direction. Also, this paper does not deal with issue of collision avoidance among agents.

\subsection{Notations}
We introduce a few additional notations, which are used in this paper. We use the bold face letters $\pmb{r} = (r_1, \ldots, r_N)^T \in \mathbb{C}^N$, $\pmb{\theta} = (\theta_1, \ldots, \theta_N)^T \in \mathbb{T}^N$, where $\mathbb{T}^N$ is the $N$-torus, which is equal to $S^1 \times \ldots \times S^1$ ($N$-times) to represent the vectors of length $N$ for the agent's positions and heading angles, respectively. Next, we define the inner product $\left<z_1, z_2\right>$ of the two complex numbers $z_1, z_2 \in \mathbb{C}$ as $\left<z_1, z_2\right> = \textrm{Re}(\bar{z}_1z_2)$, where $\bar{z}_1$ represents the complex conjugate of $z_1$. For vectors, we use the analogous boldface
notation $\left<\pmb{w}, \pmb{z}\right> = \textrm{Re}(\pmb{w^*}\pmb{z})$ for $\pmb{w}, \pmb{z} \in \mathbb{C}^N$, where $\pmb{w^*}$ denotes the conjugate transpose of $\pmb{w}$. The norm of $\pmb{z} \in \mathbb{C}^N$ is defined as $\|\pmb{z}\| = \left<\pmb{z}, \pmb{z}\right>^{1/2}$. The vectors $\pmb{0}$ and $\pmb{1}$ are used to represent by $\pmb{0} = (0,0,\ldots,0)^T \in \mathbb{R}^N$ and $\pmb{1} = (1,1,\ldots,1)^T \in \mathbb{R}^N$, respectively.

\subsection{Background}
At first, our main focus is to seek a feedback control $u_k$ for all $k$, such that the collective motion of all the agents is stabilized to a synchronized formation. The control over the average linear momentum of the group of agents is a mean to achieve synchronized formation. The average linear momentum $p_\theta$ of the group of agents, which is also referred to as the phase order parameter \cite{Strogatz2000}, is
\begin{equation}
\label{phase_order_parameter}p_\theta = \left|p_\theta\right|e^{i\Psi} = \frac{1}{N}\sum_{k=1}^{N}e^{i\theta_k}.
\end{equation}
Note that since $\theta_k, \forall k$, is a function of time, $p_\theta$ varies with time.

From \eqref{phase_order_parameter}, it is clear that the magnitude of $p_\theta$ satisfies $0\leqslant \left|{p_\theta}\right|\leqslant{1}$. In synchronized formation, the heading angles $\theta_k$, for all $k$, are the same and, hence, all the agents move in a common direction. It turns out that the phase arrangement $\pmb{\theta}$ is synchronized if $|p_\theta| = 1$.

We deal with two cases of interaction networks among agents: $i)$ all-to-all communication topology$-$ in which each agent can communicate with all other agents of the group. $ii)$ limited communication topology$-$ in which an agent can communicate with certain number of neighbors (which is also a more general case of $(i)$). We assume that limited communication topology among agents is undirected and time-invariant.

In order to achieve synchronization of the agents, the control $u_k,~\forall k$, is now proposed separately for both of these communication scenarios.

\subsubsection{All-to-All Communication Topology}
It is evident from the above discussion that the stabilization of synchronized formation can be accomplished by considering the following potential function
\begin{equation}
\label{potential function}U(\pmb{\theta}) = \frac{N}{2}\left(1 - |p_\theta|^2\right),
\end{equation}
which is minimized when $|p_\theta| = 1$, that is, when all the phases are identical (synchronized). Since $0 \leq |p_\theta| \leq 1$, the potential $U(\pmb{\theta})$ satisfies $0 \leq U(\pmb{\theta}) \leq N/2$.

Now, we state the following theorem, which says that a Lyapunov-based control framework exists to stabilize synchronized formation.

\begin{thm}\label{Theorem1}
Consider the system dynamics \eqref{modelNew} with the control law
\begin{equation}
\label{control1}u_k = K_k\left(\frac{\partial U}{\partial \theta_k}\right);~~K_k \neq 0,
\end{equation}
and define a term
\begin{equation}
\label{term}T_k(\pmb{\theta}) = \left(\frac{\partial U}{\partial \theta_k}\right)^2
\end{equation}
for all $k = 1, \ldots, N$. If $\sum_{k=1}^{N} K_k T_k(\pmb{\theta}) < 0$, all the agents asymptotically stabilize to a synchronized formation. Moreover, $K_k < 0$ for all $k$, is a restrictive sufficient condition in stabilizing synchronized formation.
\end{thm}

\begin{proof}
Consider the potential function $U(\pmb{\theta})$, defined by \eqref{potential function}, the minimization of which leads to a synchronized formation. Since the magnitude of the average linear momentum $|p_\theta|$ in \eqref{phase_order_parameter} satisfies $0 \leqslant |p_\theta| \leqslant 1$, it ensures that $U(\pmb{\theta}) \geqslant 0$. Also, $U(\pmb{\theta}) = 0$ only at the equilibrium point where $|p_\theta| = 1$. Thus, $U(\pmb{\theta})$ is a Lyapunov function candidate \cite{Khalil2000}.

The time derivative of $U(\pmb{\theta})$ along the dynamics \eqref{modelNew} is
\begin{equation}
\label{U_dot}\dot{U}(\pmb{\theta}) = \sum_{k=1}^{N}\left(\frac{\partial U}{\partial \theta_k}\right) \dot{\theta}_k  = \sum_{k=1}^{N}\left(\frac{\partial U}{\partial \theta_k}\right) u_k.
\end{equation}

Using \eqref{control1} and \eqref{term}
\begin{equation}
\label{Udot}\dot{U}(\pmb{\theta}) = \sum_{k=1}^{N}K_k\left(\frac{\partial U}{\partial \theta_k}\right)^2 = \sum_{k=1}^{N}K_k T_k(\pmb{\theta}).
\end{equation}
It shows that $\dot{U}(\pmb{\theta}) < 0$, if $\sum_{k=1}^{N}K_k T_k(\pmb{\theta}) < 0$. According to the Lyapunov stability theorem \cite{Khalil2000}, all the solutions of \eqref{modelNew} with the control \eqref{control1} asymptotically stabilize to the equilibrium where $U(\pmb{\theta})$ attains its minimum value, that is, at $|p_\theta| = 1$ (synchronized formation).

The restricted sufficiency condition is proved next. Note that the term $T_k(\pmb{\theta}) \geq 0$ for all $k = 1, \ldots, N$, which ensures that $\dot{U}(\pmb{\theta}) \leq  0$ for $K_k < 0,~\forall k$. Moreover, $\dot{U}(\pmb{\theta}) = 0$ if and only if $({\partial U}/{\partial \theta_k}) = 0$, that is, on the critical set of $U(\pmb{\theta})$. The critical set of $U(\pmb{\theta})$ is the set of all $\pmb{\theta} \in \mathbb{T}^N$, for which $({\partial U}/{\partial \theta_k}) = 0,~\forall k$. Note that
\begin{align}
\nonumber \frac{\partial U}{\partial \theta_k} &= -\frac{N}{2}\frac{\partial}{\partial \theta_k}\left<p_\theta, p_\theta\right>\\
\nonumber &= -\frac{N}{2}\left(\left<p_\theta, \frac{\partial p_\theta}{\partial \theta_k}\right> + \left<\frac{\partial p_\theta}{\partial \theta_k}, p_\theta\right>\right)\\
\label{note_that}&= -\left<p_\theta, ie^{i\theta_k}\right>.
\end{align}
Since $\pmb{\theta} \in \mathbb{T}^N$ is compact, it follows from the LaSalle's invariance theorem \cite{Khalil2000}, all the solutions of \eqref{modelNew} under control \eqref{control1} converge to the largest invariant set contained in $\{\dot{U}(\pmb{\theta}) = 0\}$, that is, the set
\begin{equation}
\Lambda = \left\{\pmb{\theta}~|~ ({\partial U}/{\partial \theta_k}) = -\left<p_\theta, ie^{i\theta_k}\right> = 0,~\forall k\right\},
\end{equation}
which is the critical set of $U(\pmb{\theta})$. In this set, dynamics \eqref{modelNew2} reduces to $\dot{\theta}_k = 0, \forall k$, which implies that all the agents move in a straight line. The set $\Lambda$ is itself invariant since
\begin{align}
\nonumber \frac{d}{dt}\left<p_\theta, ie^{i\theta_k}\right> &= \left<p_\theta, \frac{d(ie^{i\theta_k})}{dt}\right> + \left<\frac{dp_\theta}{dt}, ie^{i\theta_k}\right>\\
&= -\left<p_\theta, e^{i\theta_k}\right>\dot{\theta}_k + \frac{1}{N}\left<\sum_{k=1}^{N} ie^{i\theta_k} \dot{\theta}_k, ie^{i\theta_k}\right> = 0
\end{align}
on this set. Therefore, all the trajectories of the system \eqref{modelNew} under control \eqref{control1} asymptotically converges to the critical set of  $U(\pmb{\theta})$. Moreover, the synchronized state characterizes the stable equilibria of the system \eqref{modelNew} in the critical set $\Lambda$ and the rest of the critical points are unstable equilibria, which is proved next.

{\it Analysis of the critical set}: The critical points of $U(\pmb{\theta})$ are given by the $N$ algebraic equations
\begin{equation}
\frac{\partial U}{\partial \theta_k} = -\left<p_\theta, ie^{i\theta_k}\right> = -|p_\theta|\sin(\Psi-\theta_k) = 0,~~1\leq k \leq N,
\end{equation}
where, $p_\theta = |p_\theta|e^{i\Psi}$, as defined in \eqref{phase_order_parameter}, has been used. Since the critical points with $p_\theta = 0$ are the global maxima of $U(\pmb{\theta})$, and hence unstable if $K_k < 0,~\forall k$.

Now, we focus on the critical points for which $p_\theta \neq 0$, and $\sin(\Psi - \theta_k) = 0, \forall k$. This implies that $\theta_k \in \{\Psi~\text{mod}~2\pi, (\Psi + \pi)~\text{mod}~2\pi\},~\forall k$. Let $\theta_k = (\Psi + \pi)~\text{mod}~2\pi$ for $k \in \{1,\ldots,M\}$, and $\theta_k = \Psi~\text{mod}~2\pi$ for $k \in \{M+1,\ldots, N\}$. The value $M=0$ defines synchronized state and corresponds to the global minimum of $U(\pmb{\theta})$. Therefore, the set of synchronized state is asymptotically stable if  $K_k < 0,~\forall k$. Every other value of $1 \leq M \leq N-1$ corresponds to the saddle point, and is, therefore, unstable for $K_k < 0,~\forall k$. This is proved below.

Let $H(\pmb{\theta}) = [h_{jk}(\pmb{\theta})]$ be the Hessian of $U(\pmb{\theta})$. Then, we can find the components $[h_{jk}(\pmb{\theta})]$ of $H(\pmb{\theta})$ by evaluating the second derivatives $\frac{\partial^2 U}{\partial\theta_j \partial\theta_k}$ for all pairs of $j$ and $k$, which yields
\[
    h_{jk}(\pmb{\theta})=
\begin{cases}
    \dfrac{1}{N} - \left<p_\theta, e^{i\theta_k}\right> = \dfrac{1}{N} -|p_\theta|\cos(\Psi-\theta_k), &  j = k\\
    \dfrac{1}{N}\left<e^{i\theta_j}, e^{i\theta_k}\right> = \dfrac{1}{N} \cos(\theta_j-\theta_k),      &  j \neq k.
\end{cases}
\]

Since $\theta_k = (\Psi + \pi)~\text{mod}~2\pi$ for $k \in \{1,\ldots,M\}$, and $\theta_k = \Psi~\text{mod}~2\pi$ for $k \in \{M+1,\ldots, N\}$, $\cos(\Psi-\theta_k) = 1$ for $k \in \{1,\ldots,M\}$, and $\cos(\Psi-\theta_k) = -1$ for $k \in \{M+1,\ldots, N\}$. Hence, the diagonal entries ($j = k$) of the Hessian $H(\pmb{\theta})$ are given by
\[
    h_{kk}(\pmb{\theta})=
\begin{cases}
    (1/N) + |p_\theta|, &  k\in\{1, \ldots, M\}\\
    (1/N) - |p_\theta|, & k\in\{M+1, \ldots, N\},
\end{cases}
\]
where, $1 \leq M \leq N-1$. Since $(1/N) + |p_\theta| > 0$, the Hessian matrix $H(\pmb{\theta})$ has at least one positive pivot, and hence one positive eigenvalue \cite{strang2007}. In order to show that all critical points $1 \leq M \leq N-1$ are saddle points, we verify that the Hessian matrix $H(\pmb{\theta})$ is indefinite by showing that it has at least one negative eigenvalue.

Since $\theta_k$ is as given above, $\cos(\theta_j-\theta_k) = 1$ for $j,k\in\{1, \ldots, M\}~\text{or}~j,k\in\{M+1, \ldots, N\}$, and $\cos(\theta_j-\theta_k) = -1$ for $j \in\{1, \ldots, M\}, k\in\{M+1, \ldots, N\}~\text{or}~j \in \{M+1, \ldots, N\}, k \in \{1, \ldots, M\}$. Hence, the off diagonal entries ($j \neq k$) of $H(\pmb{\theta})$ are given by
\[
    h_{jk}(\pmb{\theta})=
\begin{cases}
    (1/N), &  \left.
    \begin{array}{l}
      j,k\in\{1, \ldots, M\}~\text{or}~j,k\in\{M+1, \ldots, N\}
    \end{array}
  \right.\\
    -(1/N), & ~~\text{otherwise}.
\end{cases}
\]

Define a vector $\pmb{w} = [w_1, \ldots, w_M, -w_{M+1}, \ldots, -w_N]^T$, with $w_k = 1,~\forall k$. Then, the Hessian $H(\pmb{\theta})$ can be compactly written as
\begin{equation}
H(\pmb{\theta}) = \frac{1}{N}\pmb{w}\pmb{w}^T + |p_\theta|\text{diag}(\pmb{w}),
\end{equation}
where, $\text{diag}(\pmb{w})$ is a diagonal matrix whose diagonal entries are given by the entries of the vector $\pmb{w}$. Now, define a vector $\pmb{q} = [q_1, \ldots, q_N]^T$ with $q_k = 0, k = {1, \ldots, N-2}$, and $q_{N-1} = -1$ and $q_{N} = 1$. By construction, $\pmb{w}^T\pmb{q} = 0$ and it follows that
\begin{equation}
\pmb{q}^TH(\pmb{\theta})\pmb{q} = |p_\theta|\pmb{q}^T\text{diag}(\pmb{w})\pmb{q} = -2|p_\theta| < 0,
\end{equation}
which shows that $H(\pmb{\theta})$ is an indefinite matrix. Hence, the critical points for which the phase angles are not synchronized and $p_\theta \neq 0$ are the saddle points and unstable for $K_k < 0,~\forall k$. This completes the proof.
\end{proof}

If the agents move at an angular frequency $\omega_0$ around individual circular orbits, we have the following corollary to Theorem~\ref{Theorem1}, which ensures the stabilization of their synchronized formation.
\begin{cor}\label{cor1}
Under the control law given by
\begin{equation}
\label{control2}u_k = \omega_0 + K_k\left(\frac{\partial U}{\partial \theta_k}\right);~~K_k \neq 0,
\end{equation}
for all $k=1, \ldots, N$, the conclusions of Theorem~\ref{Theorem1} are same for the system dynamics \eqref{modelNew}.
\end{cor}

\begin{proof}
Under the control \eqref{control2}, the time derivative of $U(\pmb{\theta})$ along the dynamics \eqref{modelNew} is
\begin{equation}
\label{U_dot_New}\dot{U}(\pmb{\theta}) = \omega_0\sum_{k=1}^{N}\left(\frac{\partial U}{\partial \theta_k}\right) + \sum_{k=1}^{N}K_k\left(\frac{\partial U}{\partial \theta_k}\right)^2
\end{equation}

From \eqref{note_that}, we note that
\begin{equation}
\label{relation}\sum_{k=1}^{N}\frac{\partial U}{\partial \theta_k} = -\sum_{k=1}^{N}\left<p_\theta, ie^{i\theta_k}\right> = -\frac{1}{N}
\sum_{k=1}^{N}\sum_{\substack{j=1, \\ j \neq k}}^{N} \sin(\theta_j - \theta_k) = 0.
\end{equation}

Using \eqref{relation}, \eqref{U_dot_New} can be rewritten as
\begin{equation}
\label{Udot_final}\dot{U}(\pmb{\theta}) = \sum_{k=1}^{N}K_k T_k(\pmb{\theta}),
\end{equation}
which is the same as \eqref{Udot}. Therefore, the conclusions of Theorem~\ref{Theorem1} are unchanged under control \eqref{control2}.
\end{proof}

\subsubsection{Limited Communication Topology}
At first, we introduce a few terms pertaining to limited communication topology which will be useful in the framework of this paper.

A graph is a pair $\mathcal{G} = (\mathcal{V}, \mathcal{E})$, where $\mathcal{V} = \{v_1, \ldots, v_N\}$ is a set of $N$ nodes or vertices and $\mathcal{E} \subseteq \mathcal{V}\times \mathcal{V}$ is a set of edges or links. Elements of $\mathcal{E}$ are denoted as $(v_j, v_k)$ which is termed an edge or a link from $v_j$ to $v_k$. A graph $\mathcal{G}$ is called an undirected graph if it consists of only undirected links. The node $v_j$ is called a neighbor of node $v_k$ if the link $(v_j, v_k)$ exists in the graph $\mathcal{G}$. In this paper, the set of neighbors of node $v_j$ is represented by $\mathcal{N}_j$. A complete graph is an undirected graph in which every pair of nodes is connected, that is, $(v_j, v_k) \in \mathcal{E}$, $\forall j, k \in N$. The Laplacian of a graph $\mathcal{G}$, denoted by $L = [l_{jk}] \in \mathbb{R}^{N\times N}$, is defined as
\[
    l_{jk}=
\begin{dcases}
    |\mathcal{N}_j|, & \text{if } j = k\\
    -1,              & \text{if } k \in \mathcal{N}_j\\
    0                & \text{otherwise}
\end{dcases}\label{laplacian}
\]
where, $|\mathcal{N}_j|$ is the cardinality of the set $\mathcal{N}_j$. Some of the important properties of the Laplacian which are relevant to this paper can be found in \cite{Bai2011}, and are as follows: The Laplacian $L$ of an undirected graph $\mathcal{G}$ is $(\text{P}1)$ symmetric and positive semi-definite, and $(\text{P}2)$ has an eigenvalue of zero associated with the eigenvector $\pmb{1}$, that is, $L\pmb{x} = 0$ iff $\pmb{x} = \pmb{1} x_0$.

In order to account for limited communication among agents, we modify the potential function \eqref{potential function} in the following manner \cite{Sepulchre2008}:

Let $P = I_N - (1/N)\pmb{1}\pmb{1}^T$, where, $I_N$ is an $N\times N$-identity matrix, be a projection matrix which satisfies $P^2 = P$. Let the vector $e^{i\pmb{\theta}}$ be represented by $e^{i\pmb{\theta}} = (e^{i\theta_1}, \ldots, e^{i\theta_N})^T \in \mathbb{C}^N$. Then, $Pe^{i\pmb{\theta}} = e^{i\pmb{\theta}} - p_\theta\pmb{1}$. One can obtain the equality
\begin{equation}
\label{identity}||P e^{i\pmb{\theta}}||^2 = \left<e^{i\pmb{\theta}}, P e^{i\pmb{\theta}}\right> = N(1 - |p_\theta|^2),
\end{equation}
which is minimized when $|p_\theta| = 1$ (synchronized formation). Since, $P$ is ($1/N$) times the Laplacian of the complete graph, the identity \eqref{identity} suggests that the optimization of $U(\pmb{\theta})$ in \eqref{potential function} may be replaced by the optimization of
\begin{equation}
\label{laplacian_potential}W_L(\pmb{\theta}) = Q_L(e^{i\pmb{\theta}}) = ({1}/{2})\left<e^{i\pmb{\theta}}, Le^{i\pmb{\theta}}\right>,
\end{equation}
which is a Laplacian quadratic form associated with $L$. Note that, for a connected graph, the quadratic form \eqref{laplacian_potential} is positive semi-definite, and vanishes only when $e^{i\pmb{\theta}} = e^{i\theta_c}\pmb{1}$, where $\theta_c \in S^1$ is a constant (see property $\text{P}2$), that is, the potential $W_L(\pmb{\theta})$ is minimized in the synchronized formation.

\begin{thm}\label{Theorem2}
Let $L$ be the Laplacian of an undirected and connected graph $\mathcal{G} = (\mathcal{V}, \mathcal{E})$ with $N$ vertices. Consider the system dynamics \eqref{modelNew} with the control law
\begin{equation}
\label{control4}u_k = K_k\left(\frac{\partial W_L}{\partial \theta_k}\right);~~K_k \neq 0,
\end{equation}
and define a term
\begin{equation}
\label{term_new}\overline{T}_k(\pmb{\theta}) = \left(\frac{\partial W_L}{\partial \theta_k}\right)^2
\end{equation}
for all $k = 1, \ldots, N$. If $\sum_{k=1}^{N} K_k \overline{T}_k(\pmb{\theta}) < 0$, all the agents asymptotically stabilize to a synchronized formation. Moreover, $K_k < 0$ for all $k$, is a restrictive sufficient condition in stabilizing synchronized formation.
\end{thm}

\begin{proof}
The time derivative of $W_L(\pmb{\theta})$, along the dynamics \eqref{modelNew}, is
\begin{equation}
\label{W_Ldot}\dot{W}_L(\pmb{\theta}) = \sum_{k=1}^{N} \left(\frac{\partial W_L}{\partial \theta_k}\right) \dot{\theta}_k = \sum_{k=1}^{N} \left(\frac{\partial W_L}{\partial \theta_k}\right)u_k.
\end{equation}
Using \eqref{control4} and \eqref{term_new}
\begin{equation}
\label{WLdot}\dot{W}_L(\pmb{\theta}) = \sum_{k=1}^{N} K_k\left(\frac{\partial W_L}{\partial \theta_k}\right)^2 = \sum_{k=1}^{N} K_k\overline{T}_k.
\end{equation}
The rest of the proof proceeds in a similar way as the proof of Theorem~\ref{Theorem1}. We just need to analyze the critical set of the potential $W_L(\pmb{\theta})$, which is as follows:

{\it Analysis of the critical set}:
The critical set of $W_L(\pmb{\theta})$ is the set of all $\pmb{\theta} \in \mathbb{T}^N$ for which $(\partial W_L/\partial \theta_k) = 0,~\forall k$. Note that
\begin{equation}
\label{note_that_new}\frac{\partial W_L}{\partial \theta_k} = \frac{1}{2}\sum_{j=1}^{N}\frac{\partial}{\partial \theta_k}\left<e^{i\theta_j}, L_je^{i\pmb{\theta}}\right> = \left<ie^{i\theta_k}, L_ke^{i\pmb{\theta}}\right>,
\end{equation}
where, $L_k$ is the $k^\text{th}$ row of the Laplacian $L$. Thus, the critical points of $W_L(\pmb{\theta})$ are given by the $N$ algebraic equations
\begin{equation}
\frac{\partial W_L}{\partial \theta_k} = \left<ie^{i\theta_k}, L_ke^{i\pmb{\theta}}\right> = 0,~~~1 \leq k \leq N.
\end{equation}

Let $e^{i\bar{\pmb{\theta}}}$ be an eigenvector of $L$ with eigenvalue $\lambda \in \mathbb{R}$. Then, $Le^{i\bar{\pmb{\theta}}} = \lambda e^{i\bar{\pmb{\theta}}}$, and
\begin{equation}
\label{critical_set_1}\frac{\partial W_L}{\partial \theta_k}\Big{|}_{\pmb{\theta} = \bar{\pmb{\theta}}} = \left<ie^{i\bar{\theta}_k}, L_ke^{i\bar{\pmb{\theta}}}\right> = \lambda\left<ie^{i\bar{\theta}_k}, e^{i\bar{\theta}_k}\right> = 0,
\end{equation}
which implies that $\bar{\pmb{\theta}}$ is a critical point of $W_L(\pmb{\theta})$. Since graph $\mathcal{G}$ is undirected, the Laplacian $L$ is symmetric, and hence its eigenvectors associated with distinct eigenvalues are mutually orthogonal \cite{strang2007}. Since $\mathcal{G}$ is also connected, $\pmb{1}$ spans the kernel of $L$. Therefore, the eigenvector associated with $\lambda = 0$ is $e^{i\bar{\pmb{\theta}}} = e^{i\theta_c}\pmb{1}$ for any $\theta_c \in S^1$, which implies $\bar{\pmb{\theta}}$ is synchronized. All the remaining eigenvectors satisfy $\pmb{1}^Te^{i\bar{\pmb{\theta}}} = 0$ and characterize the unstable equilibria. This completes the proof.
\end{proof}

Similar to Corollary~\ref{cor1}, we have the following corollary to Theorem~\ref{Theorem2}.
\begin{cor}\label{cor2}
Let $L$ be the Laplacian of an undirected and connected graph $\mathcal{G} = (\mathcal{V}, \mathcal{E})$ with $N$ vertices. Under the control law given by
\begin{equation}
\label{control6}u_k = \omega_0 + K_k\left(\frac{\partial W_L}{\partial \theta_k}\right);~~K_k \neq 0,
\end{equation}
for all $k=1, \ldots, N$, the conclusions of Theorem~\ref{Theorem2} are same for the system dynamics \eqref{modelNew}.
\end{cor}

\begin{proof}
Under the control \eqref{control6}, the time derivative of $W_L(\pmb{\theta})$ along the dynamics \eqref{modelNew} is
\begin{equation}
\label{WL_dot_New}\dot{W_L}(\pmb{\theta}) = \omega_0\sum_{k=1}^{N}\left(\frac{\partial W_L}{\partial \theta_k}\right) + \sum_{k=1}^{N}K_k\left(\frac{\partial W_L}{\partial \theta_k}\right)^2
\end{equation}

From \eqref{note_that_new}, we note that
\begin{equation}
\label{relation_limited}\sum_{k=1}^{N}\frac{\partial W_L}{\partial \theta_k} = \sum_{k=1}^{N}\left<ie^{i\theta_k}, L_ke^{i\pmb{\theta}}\right> = -\sum_{k=1}^{N}\sum_{j \in \mathcal{N}_k} \sin(\theta_j - \theta_k) = 0.
\end{equation}

Using \eqref{relation_limited}, \eqref{WL_dot_New} can be rewritten as
\begin{equation}
\label{WLdot_final}\dot{W_L}(\pmb{\theta}) = \sum_{k=1}^{N}K_k \overline{T}_k(\pmb{\theta}),
\end{equation}
which is the same as \eqref{WLdot}. Therefore, the conclusions of Theorem~\ref{Theorem2} are unchanged under control \eqref{control6}.
\end{proof}

\subsection{Problem Description}
Now, we formally state the main objective of this paper. Using \eqref{relation} and \eqref{relation_limited}, the control laws, given by \eqref{control1}, and \eqref{control4} can be written as
\begin{equation}
\label{control3}\dot{\theta}_k = - \frac{K_k}{N} \sum_{\substack{j=1, \\ j \neq k}}^{N} \sin(\theta_j - \theta_k),
\end{equation}
\begin{equation}
\label{control5}\dot{\theta}_k = - K_k \sum_{j \in \mathcal{N}_k} \sin(\theta_j - \theta_k),
\end{equation}
for the all-to-all and limited communication scenarios, respectively. The term $K_k$ in the control laws \eqref{control3} and \eqref{control5} is the controller gain for the $k^\text{th}$ agent. Prior work in \cite{Sepulchre2007} uses the same controller gain $K$ for all $k$, whereas we extend the analysis by using different gains $K_k$ for different agents. This is the heterogeneous controller gains case of interest in this paper. In subsequent sections, we will explore the effect of heterogeneous controller gains on the final velocity direction of the agents in their synchronized formation.

\begin{remark}
Note that, in the Theorem~\ref{Theorem1} (Theorem~\ref{Theorem2}), the conditions $\sum_{k=1}^{N}K_k T_k(\pmb{\theta}) < 0 \left(\sum_{k=1}^{N}K_k \overline{T}_k(\pmb{\theta})< 0 \right) $ may be satisfied for both positive and negative values of gains $K_k$ because of the involvement of the term $T_k(\pmb{\theta}) (\overline{T}_k(\pmb{\theta}))$. However, in this paper, the idea of introducing heterogeneous gains is illustrated mainly for the restrictive sufficient condition on $K_k$, that is, $K_k < 0, \forall k$, since the analysis for the set of gains $K_k$ satisfying $\sum_{k=1}^{N}K_k T_k(\pmb{\theta}) < 0 \left(\sum_{k=1}^{N}K_k \overline{T}_k(\pmb{\theta})< 0 \right)$ is quite involved for $N > 2$. Moreover, it will be shown for the simple case of $N=2$ that the reachable set of the final velocity direction further expands for the controller gains $K_k$ satisfying the condition $\sum_{k=1}^{N}K_k T_k(\pmb{\theta}) < 0 \left(\sum_{k=1}^{N}K_k \overline{T}_k(\pmb{\theta})< 0 \right)$.
\end{remark}

\begin{figure*}
\centering
\subfigure[$e^{i\theta_{k0}} \in S$]{\includegraphics[scale=0.98]{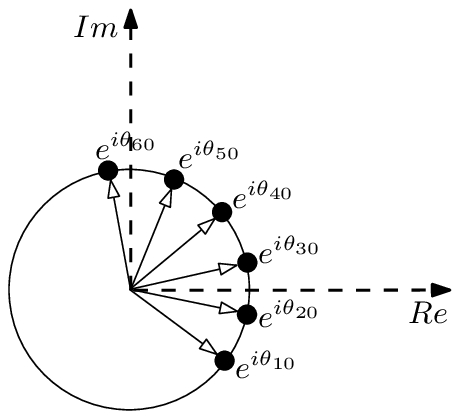}}\hspace{0.4cm}
\subfigure[$Co(S)$]{\includegraphics[scale=0.98]{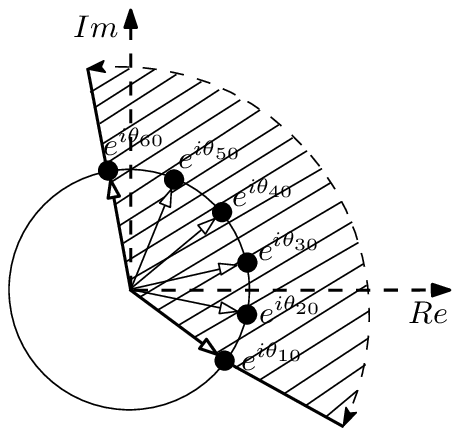}}\hspace{0.4cm}
\subfigure[$Co(S) \bigcap S_z$]{\includegraphics[scale=0.98]{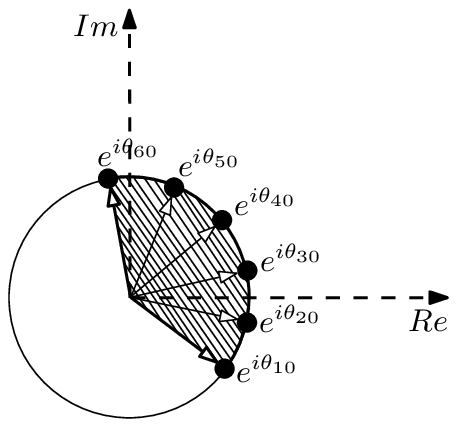}}
\caption{Arrangement of all the initial vectors around the unit circle for $N=6$. $(a)$ All the unit vectors $e^{i\theta_{k0}},~k = 1, \ldots, 6$, belong to the set $S$. $(b)$ Conic hull of the arrangement of these unit vectors $e^{i\theta_{k0}},~k = 1, \ldots, 6$. $(c)$ Region $Co(S) \bigcap S_z$.}
\label{Arrangement of all the points}
\end{figure*}

\section{Synchronized Formation and Reachable Velocity Directions}
The agents are said to be in synchronized formation when, at all times, the direction of their movement approaches a common velocity direction $\theta_c \in S^1$, that is,
\begin{equation}
\label{angle_relation_synchronization}\theta_1(t) = \theta_2(t) = \theta_3(t) =,\ldots, \theta_N(t) = \theta_c~(\text{mod}~2\pi).
\end{equation}

At first, we derive an analytical expression of $\theta_c$ for $\omega_0 = 0$. Then, we extend these results to $\omega_0 \neq 0$ by performing the analysis in a rotating frame of reference.

\subsection{Case~1: $\omega_0 = 0$}
For $\omega_0 = 0$, synchronization corresponds to parallel motion of all the agents in a fixed direction $\theta_c$, with arbitrary but constant relative spacing.

Before proceeding further, we state the following definitions from \cite{Barvinok2002}$-$\cite{Nef1988}, based on which further analysis is carried out.

\begin{defn}\label{convex_cone}
(cone, convex cone and conic hull) Let $V$ be a vector space. A set $\Gamma\subset V$ is called a cone if $\pmb{0} \in \Gamma$ and $\lambda \pmb{x} \in \Gamma$ for every $\lambda \geq 0$ and every $\pmb{x} \in \Gamma$. Moreover, the set $\Gamma \subset V$ is called a convex cone if $\pmb{0} \in \Gamma$ and if for any two points $\pmb{x}, \pmb{y} \in \Gamma$ and any two numbers $a, b \geq 0$, the point $\pmb{z} = a\pmb{x} + b\pmb{y}$ is also in $\Gamma$. Given points $\pmb{x_1}, \ldots, \pmb{x_m} \in \Gamma$ and non-negative numbers $\tau_1, \ldots, \tau_m$, the point
\begin{equation}
\pmb{x} = \sum_{j=1}^{m} \tau_j \pmb{x}_j
\end{equation}
is called a conic combination of the points $\pmb{x_1}, \ldots, \pmb{x_m}$. The set $Co(S)$ of all conic combinations from a set $S \subset \Gamma$ is called the conic hull of the set $S$.
\end{defn}

\begin{defn}
(ray and extreme ray) Let $V$ be a vector space and the set $\Gamma \subset V$ be a cone. The set of points $\lambda \pmb{x}$, $\lambda \geq 0$ of a non-zero point $\pmb{x} \in \Gamma$ is called a ray spanned by $\pmb{x}$. Let $\Gamma_1 \subset \Gamma$ be a ray. We say that $\Gamma_1$ is an extreme ray of $\Gamma$ if for any $\pmb{v} \in \Gamma_1$ and any $\pmb{x}, \pmb{y} \in \Gamma$, whenever $\pmb{v} = (\pmb{x} + \pmb{y})/2$, we must have  $\pmb{x}, \pmb{y} \in \Gamma_1$.
\end{defn}

\begin{defn}
(acute convex cone) A convex cone $\Gamma$ is said to be an acute convex cone if $\Gamma \bigcap (-\Gamma) = \{0\}$, that is, if $\pmb{x} \in \Gamma$ and $-\pmb{x} \in \Gamma$ implies $\pmb{x} = \pmb{0}$.
\end{defn}

Based on these definitions, we further define the following terms useful in the framework of this paper.

Let the agents, with dynamics given by \eqref{modelNew}, start from initial heading angles $\pmb{\theta}(0) = (\theta_{10}, \ldots, \theta_{N0})^T \in \mathbb{D}^N$, where $\mathbb{D} = (-\pi, \pi)$. Let us define $S = \left\{e^{i\theta_{k0}},~k = 1, \ldots, N\right\}$ as the set of points around the unit circle in the complex plane and let $Co(S)$ be the conic hull of $S$. For $N=6$, Fig.~$1(a)$ shows one of the arrangements of all the unit vectors belonging to the set $S$, and for this arrangement, $Co(S)$ is shown by the shaded region in Fig.~$1(b)$. In Fig.~$1(b)$, $e^{i\theta_{10}}$ and $e^{i\theta_{60}}$ are the unit vectors along the extreme rays of $Co(S)$.

Let $S_z = \{z \in \mathbb{C}~{\big |}~|z| \leq 1\}$ be the set of all the points residing in the interior and on the boundary of a unit circle in the complex plane. Then, $Co(S) \bigcap S_z$ is a circular sector as shown by the shaded region in Fig.~$1(c)$.

Based on these notations, we now state the following lemma which depicts the behavior of the order parameter $p_\theta$ with time against heterogeneous controller gains $K_k < 0, \forall k$.

\begin{lem}\label{lem1}
Consider $N$ agents, with dynamics given by \eqref{modelNew}, under the control law \eqref{control3} with $K_k < 0, \forall k$. Let the initial heading angle of the agents be given by $\pmb{\theta}(0)$ such that $Co(S)$ is an acute convex cone. Then,
\begin{equation}
p_\theta \in Co(S) \bigcap S_z,~\forall t\geq0,
\end{equation}
where, $p_\theta$ is the order parameter, and is defined by \eqref{phase_order_parameter}.
\end{lem}

\begin{figure}[!t]
\centering
\includegraphics[scale=0.67]{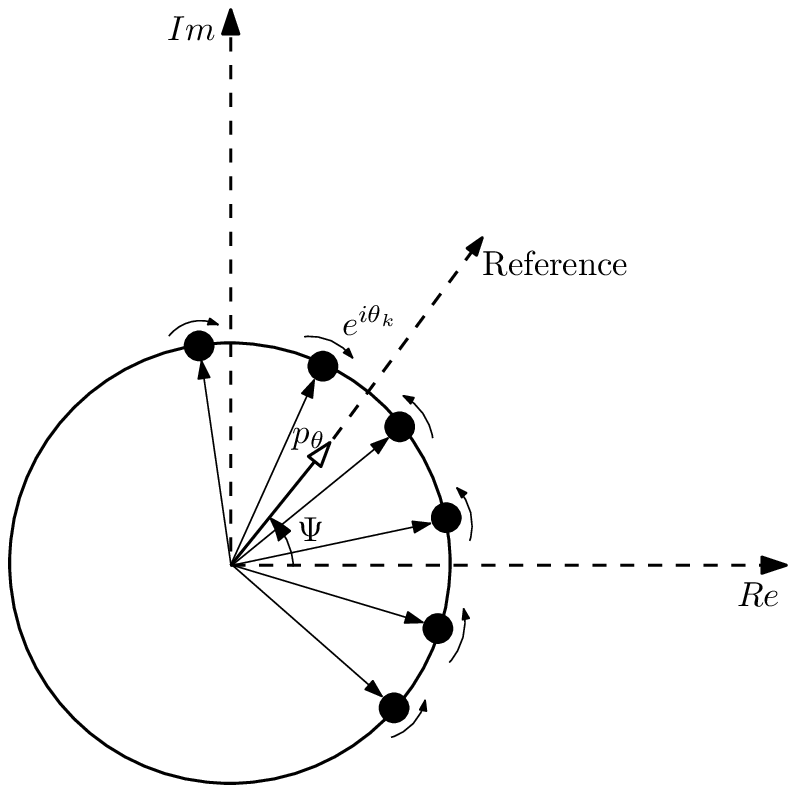}
\caption{The unit vectors $e^{i\theta_{k}}$ at a particular instant of time $t = t_1$. All the vectors are pulled toward the average phase $\Psi$ of the order parameter $p_\theta$.}
\label{Dynamics of all the unit vectors}
\end{figure}

\begin{proof}
From \eqref{phase_order_parameter}, we can write
\begin{equation}
\left|p_\theta\right|e^{i(\Psi - \theta_k)} =  \frac{1}{N}\sum_{\substack{j=1, \\ j \neq k}}^{N}e^{i(\theta_j - \theta_k)},
\end{equation}
the imaginary part of which is given by
\begin{equation}
\label{phase_order_parameter_New}\left|p_\theta\right|\sin(\Psi - \theta_k) = \frac{1}{N}\sum_{\substack{j=1, \\ j \neq k}}^{N} \sin(\theta_j - \theta_k)
\end{equation}
Using \eqref{phase_order_parameter_New}, \eqref{control3} can be written as
\begin{equation}
\label{theta_dot_new}\dot{\theta}_k = -K_k\left|p_\theta\right|\sin(\Psi - \theta_k),
\end{equation}
which implies that the heading angle $\theta_k$ of the $k^\text{th}$ agent is pulled toward the average phase $\Psi$ of the whole ensemble. The interpretation of the dynamics \eqref{theta_dot_new}, at a particular instant of time $t = t_1$, is shown in Fig.~\ref{Dynamics of all the unit vectors}, where the heading rate vectors, $e^{i\theta_k}$ of all the agents are represented as the swarm of points moving around the unit circle in the complex plane. By scaling each vector $e^{i\theta_k}$ by a factor of $1/N$, and then taking their resultant over all $k = 1, \ldots, N$, we get the vector $p_\theta$.

For better understanding of the dynamics \eqref{theta_dot_new}, $\forall k$, and $\forall t$, it is convenient to choose the reference axis along the order parameter $p_\theta$, as shown in Fig.~\ref{Dynamics of all the unit vectors}, and measure the angle of each unit vector with respect to it. By doing so, it is easy to see that $\left|\Psi - \theta_k\right| < \pi$ for $e^{i\theta_{k0}} \in S, \forall k$. Therefore, for $K_k < 0, \forall k$, one can observe from \eqref{theta_dot_new} that, if $0 < \Psi - \theta_k < \pi$ (that is, for unit vectors lying in the clockwise direction of $p_\theta$), $\dot{\theta}_k > 0$, and if $-\pi < \Psi - \theta_k < 0$ (that is, for unit vectors lying in the anticlockwise direction of $p_\theta$), $\dot{\theta}_k < 0$. It means that the heading angle of the $k^\text{th}$ agent always pulls toward the average phase $\Psi$ of the group.

Also, at time instant $t=0$, the linear momentum vector $p_\theta$ from \eqref{phase_order_parameter} is given by
\begin{equation}
p_{\theta}(0) = \sum_{k=1}^{N} \mu_k e^{i\theta_{k0}},
\end{equation}
where, $\mu_k = {1}/{N}, \forall k$, is a non-zero constant. Since $|p_\theta(0)| \leq 1$, the vector $p_{\theta}(0)$, according to the above definitions, lies in $Co(S) \bigcap S_z$ for $e^{i\theta_{k0}} \in S, \forall k$. Moreover, since all the unit vectors $e^{i\theta_k}$, at all times, approach $p_\theta$, the order parameter $p_\theta$ always remains in $Co(S) \bigcap S_z$, that is, $p_\theta \in Co(S) \bigcap S_z, \forall t\geq0$. This completes the proof.
\end{proof}

The previous result is obtained for the all-to-all communication scenario. Similarly, in the limited communication scenario, by using the phase order parameter
\begin{equation}
p^k_\theta = \frac{1}{N}\sum_{j \in \mathcal{N}_k}e^{i\theta_j} = |p^k_\theta|e^{i\Psi^k},
\end{equation}
it can be proved that the $k^\text{th}$ agent always approaches to vector $p^k_\theta$. Since, $\forall k$, $p^k_\theta \in Co(S) \bigcap S_z,~\forall t\geq0$, all the agents synchronize at an angle within the acute convex cone only.

\begin{figure}[!t]
\centering
\includegraphics[scale= 1.2]{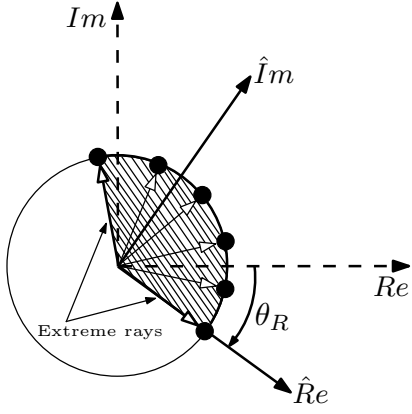}
\caption{The new coordinate system, obtained by rotating the standard coordinate system by an angle $\theta_R \in [-\pi, \pi)$. The angle $\theta_R$ is chosen such that the real axis of this new coordinate system lies along that extreme ray of $Co(S)$ so that all the initial heading angles $\pmb{\theta}(0)$, in this new coordinate system, are non-negative.}
\label{New coordinated system}
\end{figure}

Note that, depending on the initial heading angle $\pmb{\theta}(0)$, the circular sector $Co(S) \bigcap S_z$, as shown in Fig.~$1(c)$, can lie anywhere in $S_z$. Thus, for the sake of convenience and without loss of generality, a new coordinate system, as shown in Fig.~\ref{New coordinated system}, is defined by rotating the standard coordinate system by an angle $\theta_R \in [-\pi, \pi)$, which is chosen such that the real axis of this new coordinate system lies along that extreme ray of $Co(S)$ which will ensure that all the initial heading angles $\pmb{\theta}(0)$, in this new coordinate system, are non-negative (measured anti-clockwise from the new reference). Thus, in the new coordinates, we have
\begin{equation}
\label{transformation}\hat{\theta}_k = \theta_k - \theta_R
\end{equation}
as the heading angle of the $k^\text{th}$ agent.

Now, we state the following theorem, in which an expression for the reachable velocity direction $\theta_c$, is obtained.
\begin{thm}\label{Theorem3}
Consider $N$ agents, with dynamics given by \eqref{modelNew}, under the control law \eqref{control3} with $K_k < 0, \forall k$. The final velocity directions of all the agents having their initial heading angles $\pmb{\theta}(0)$, such that $Co(S)$ is an acute convex cone, converge to a common value $\theta_c$ given by
\begin{equation}
\label{theta_c}\theta_c = \left\{\left(\displaystyle\sum_{k=1}^{N} \dfrac{\hat{\theta}_{k0}}{K_k}\right){\Big /}\left({\displaystyle\sum_{k=1}^{N} \dfrac{1}{K_k}}\right)\right\} + \theta_R,
\end{equation}
where, $\hat{\theta}_{k0} = \theta_{k0} - \theta_R$ is the initial heading angle of the $k^\text{th}$ agent with respect to the new coordinate system, and $\theta_c$ is called a reachable velocity direction in the synchronized formation for this system of $N$ agents.
\end{thm}

\begin{proof}
Taking the summation on both sides of \eqref{control3} over all $k=1, \ldots, N$, we get
\begin{equation}
\label{angle_relation}\sum_{k=1}^{N} \frac{\dot{\theta}_k(t)}{K_k} = -\frac{1}{N}\sum_{k=1}^{N}\sum_{\substack{j=1, \\ j \neq k}}^{N} \sin(\theta_j - \theta_k) = 0.
\end{equation}
Integration of \eqref{angle_relation} yields
\begin{equation}
\label{angle_relation1}\sum_{k=1}^{N} \frac{{\theta}_k(t)}{K_k} = \sum_{k=1}^{N} \frac{\theta_{k0}}{K_k},~~\forall t.
\end{equation}

Note that, in the new coordinates, the condition \eqref{angle_relation_synchronization} becomes
\begin{equation}
\label{angle_relation_synchronization_new}\hat{\theta}_1(t) = \hat{\theta}_2(t) = \hat{\theta}_3(t) =,\ldots, \hat{\theta}_N(t) = \hat{\theta}_c,
\end{equation}
where, Lemma~\ref{lem1} has been used to eliminate modulo $2\pi$ operation.

On substituting \eqref{angle_relation_synchronization_new} for all $k =1, \ldots, N$,  in \eqref{angle_relation1}, we get
\begin{equation}
\label{theta_c_new}\hat{\theta}_c = \left(\displaystyle\sum_{k=1}^{N} \dfrac{\hat{\theta}_{k0}}{K_k}\right){\Big /}\left({\displaystyle\sum_{k=1}^{N} \dfrac{1}{K_k}}\right),
\end{equation}
which is the reachable velocity direction with respect to the new coordinate system. Now, using transformation \eqref{transformation}, we get \eqref{theta_c} in the standard coordinate system.
\end{proof}

From \eqref{control5},
\begin{equation}
\label{angle_relation_limited}\sum_{k=1}^{N} \frac{\dot{\theta}_k(t)}{K_k} = -\sum_{k=1}^{N}\sum_{j \in \mathcal{N}_k} \sin(\theta_j - \theta_k) = 0,
\end{equation}
which implies that the result obtained in Theorem~\ref{Theorem3} (main result) also holds for the limited communication scenario. Now, based on Theorem~\ref{Theorem3}, we further obtain a few interesting results which equally hold for the limited communication scenario unless otherwise stated.

For the sake of simplicity, further analysis in this paper is carried out in the new coordinate system as shown in Fig.~\ref{New coordinated system}, which can be easily transformed to the standard coordinate system by using the transformation \eqref{transformation}.

\begin{cor}\label{cor3}
For the conditions given in Theorem~\ref{Theorem3}, the reachable velocity direction $\hat{\theta}_c$ given by \eqref{theta_c_new} is a convex combination of all the initial heading angles $\hat{\theta}_{k0}, \forall k$.
\end{cor}

\begin{proof}
Equation \eqref{theta_c_new} can also be rewritten as
\begin{equation}
\label{theta_c_complexHull}\hat{\theta}_c = \displaystyle\sum_{k=1}^{N} \left\{\left(\displaystyle\frac{1}{K_k}\right){\Big /}\left({\displaystyle\sum_{j=1}^{N} \frac{1}{K_j}}\right)\right\} \hat{\theta}_{k0}
\end{equation}
Assume that for all $k = 1, \ldots, N$,
\begin{equation}
\label{lambda}\lambda_k = \left(\displaystyle\frac{1}{K_k}\right){\Big /}\left(\displaystyle\sum_{j=1}^{N} \frac{1}{K_j}\right)
\end{equation}
Since $K_k < 0$ for all $k = 1, \ldots, N$, hence $\lambda_k > 0$ and $\sum_{k=1}^{N} \lambda_k = 1$. Substituting \eqref{lambda} in \eqref{theta_c_complexHull}, we get
\begin{equation}
\label{theta_c_complex_combination}\hat{\theta}_c = \displaystyle\sum_{k=1}^{N} \lambda_k\hat{\theta}_{k0},
\end{equation}
which shows that $\hat{\theta}_c$ is a convex combination of $\hat{\theta}_{k0}, \forall k$.
\end{proof}

\begin{cor}\label{cor4}
Let $\hat{\theta}_{m0} = \min_{k} \{\hat{\theta}_{k0}\} (= 0^\circ$ in the new coordinate system$)$ and $\hat{\theta}_{M0} = \max_{k}\{\hat{\theta}_{k0}\}$ be the angles corresponding to the extreme rays of $Co(S)$ under the conditions given in Theorem~\ref{Theorem3}. These angles are not reachable in the synchronized formation of $N$ agents.
\end{cor}

\begin{proof}
This can be proved by contradiction. Let us assume that $\hat{\theta}_{m0}$ is reachable. It means that $\exists~K_k < 0, \forall k$, such that \eqref{theta_c_new} is satisfied. Hence, from \eqref{theta_c_new}, we can write
\begin{equation}
\label{theta_m}\hat{\theta}_{m0} = \left(\displaystyle\sum_{k=1}^{N} \dfrac{\hat{\theta}_{k0}}{K_k}\right){\Big /}\left(\displaystyle\sum_{k=1}^{N} \dfrac{1}{K_k}\right),
\end{equation}
From which
\begin{equation}
\label{theta_m_1}\displaystyle\sum_{\substack{
   k = 1, \\
   k \neq m
  }}^{N}
\left(\dfrac{\hat{\theta}_{k0} - \hat{\theta}_{m0}}{K_k}\right) = 0.
\end{equation}
However, since $\hat{\theta}_{m0} = \min_{k} \{\hat{\theta}_{k0}\}$, $ \hat{\theta}_{k0} -\hat{ \theta}_{m0} > 0 $, for all $k = 1, \ldots, m-1, m+1, \ldots, N$. Thus,
\begin{equation}
\displaystyle\sum_{\substack{
   k = 1, \\
   k \neq m
  }}^{N}
\left(\dfrac{\hat{\theta}_{k0} - \hat{\theta}_{m0}}{K_k}\right) < 0
\end{equation}
as $K_k < 0, \forall k$, which contradicts \eqref{theta_m_1} and hence $\hat{\theta}_{m0}$ is not reachable. Similarly, we can show that $\hat{\theta}_{M0}$ is not reachable. This completes the proof.
\end{proof}

Now, we describe the following theorem which ensures the reachability of $\hat{\theta}_c$ in \eqref{theta_c_new} against heterogeneous controller gains $K_k < 0, \forall k$.

\begin{thm}\label{Theorem4}
Consider $N$ agents, with dynamics given by \eqref{modelNew}, under the control law \eqref{control3} with $K_k < 0, \forall k$. Let the initial heading angles of the agents be given by $\pmb{\theta}(0)$ such that $Co(S)$ is an acute convex cone. A final velocity direction $\hat{\theta}_c$, given by \eqref{theta_c_new}, of all the agents is reachable iff
\begin{equation}
\label{convex_cone_new} \hat{\theta}_c \in (\hat{\theta}_{m0}, \hat{\theta}_{M0}).
\end{equation}
\end{thm}

\begin{proof}
This directly follows from Corollary~\ref{cor3} and Corollary~\ref{cor4} that the reachable velocity $\hat{\theta}_c \in (\hat{\theta}_{m0}, \hat{\theta}_{M0})$, depending upon the heterogeneous gains $K_k < 0, \forall k$. The sufficiency condition is proved as follows.

Let $\hat{\theta}_c \in (\hat{\theta}_{m0}, \hat{\theta}_{M0})$. Then, we can find $\alpha_k$ such that
\begin{equation}\label{relation2}
\sum_{k=1}^{N} \alpha_k\hat{\theta}_{k0} = \hat{\theta}_c
\end{equation}
where, $\displaystyle\sum_{k=1}^{N} \alpha_k = 1$ with $\alpha_k > 0, \forall k$. Let us define
\begin{equation}
\label{gains}K_k = {c}{/}{\alpha_k}
\end{equation}
for all $k$, where $c < 0$ is any constant. Thus, $K_k < 0, \forall k$, and $\sum_{k=1}^{N} ({1}/{K_k}) = {1}/{c}$.
Replacing $\alpha_k$ by $K_k$ in \eqref{relation2}, we get
\begin{equation}
\nonumber \hat{\theta}_c = \sum_{k=1}^{N} \left\{\dfrac{\dfrac{1}{K_k}}{\dfrac{1}{c}}\right\}\hat{\theta}_{k0} =\sum_{k=1}^{N} \left\{\left(\dfrac{\dfrac{1}{K_k}}{\displaystyle\sum_{j=1}^{N}\dfrac{1}{K_j}}\right)\hat{\theta}_{k0}\right\} = \dfrac{\displaystyle\sum_{k=1}^{N} \dfrac{\hat{\theta}_{k0}}{K_k}}{\displaystyle\sum_{k=1}^{N} \dfrac{1}{K_k}},
\end{equation}
which is the same as \eqref{theta_c_new}. This completes the proof.
\end{proof}

\begin{remark}
If we choose homogeneous controller gains, as in \cite{Sepulchre2007}, that is, $K_k = K, \forall k$, then the reachable velocity direction $\hat{\theta}_c$, by using \eqref{theta_c_new}, is given by
\begin{equation}
\label{theta_c_avg}\overline{\hat{\theta}}_c = \frac{1}{N}\displaystyle\sum_{k=1}^{N} \hat{\theta}_{k0},
\end{equation}
which is the average of all the initial heading angles of $N$ agents. Thus, by using homogeneous controller gains, only the average consensus in initial heading angles is possible. However, by using heterogeneous controller gains, we are able to expand the reachable set of the final velocity direction $\hat{\theta}_c$. In fact, the agents can be made to converge to any desired common velocity direction $\hat{\theta}_c \in (\hat{\theta}_{m0}, \hat{\theta}_{M0})$ by suitably selecting the heterogeneous gains $K_k < 0, \forall k$. These heterogeneous gains can be selected according to \eqref{gains}. We can see that these gains are not unique since none of $\alpha_k$ and $c$ need be unique. We also observe that \eqref{theta_c_new} is independent of the initial locations of the agents. Therefore, different groups of the agents, with arbitrary initial locations, but with same individual initial velocity directions, can be made to converge to the same desired direction $\hat{\theta}_c \in (\hat{\theta}_{m0}, \hat{\theta}_{M0})$.
\end{remark}

Since it is physically impossible to get the same gains for all the agents, the idea of heterogeneous controller gains was introduced. Suppose the homogeneous gains $K$ of each agent vary within certain limits while obeying all the conditions for convergence, then we have the following theorem, which tells about the deviation of the final velocity direction $\hat{\theta}_c$ from its mean value $\overline{\hat{\theta}}_c$ given by \eqref{theta_c_avg}, and comments on its reachability.

\begin{thm}\label{Theorem5}
Let there be an error of $\epsilon_k = \eta_kK$, where $0 \leq \eta_k < 1$, in the gain $K$ of the $k^\text{th}$ agent, with dynamics given by \eqref{modelNew}, under the control law \eqref{control3} with $K_k = K < 0, \forall k$. Let $\eta = \max_{k}\{\eta_k\}$ be the maximum error, and the initial heading angles of the agents be given by $\pmb{\theta}(0)$ such that $Co(S)$ is an acute convex cone. Then, in the synchronized formation of this system of $N$ agents, the perturbed final velocity direction
\begin{equation}
\label{error_angle}\hat{\theta}^p_c \in \left(\hat{\theta}_{m0}, \hat{\theta}_{M0}\right) \bigcap \left[\overline{\hat{\theta}}_c - \Delta\hat{\theta}^l_c, \overline{\hat{\theta}}_c + \Delta\hat{\theta}^u_c\right],
\end{equation}
where,
\begin{equation}
\Delta\hat{\theta}^l_c =  \left(\frac{2\eta}{1 + \eta}\right)\overline{\hat{\theta}}_c,~~\text{and}~~\Delta\hat{\theta}^u_c =  \left(\frac{2\eta}{1 - \eta}\right)\overline{\hat{\theta}}_c,
\end{equation}
are, respectively, the maximum values of the lower and upper deviations of the reachable velocity direction from its mean value $\overline{\hat{\theta}}_c$ given by \eqref{theta_c_avg}.
\end{thm}

\begin{figure*}
\centering
\subfigure[]{\includegraphics[scale=0.45]{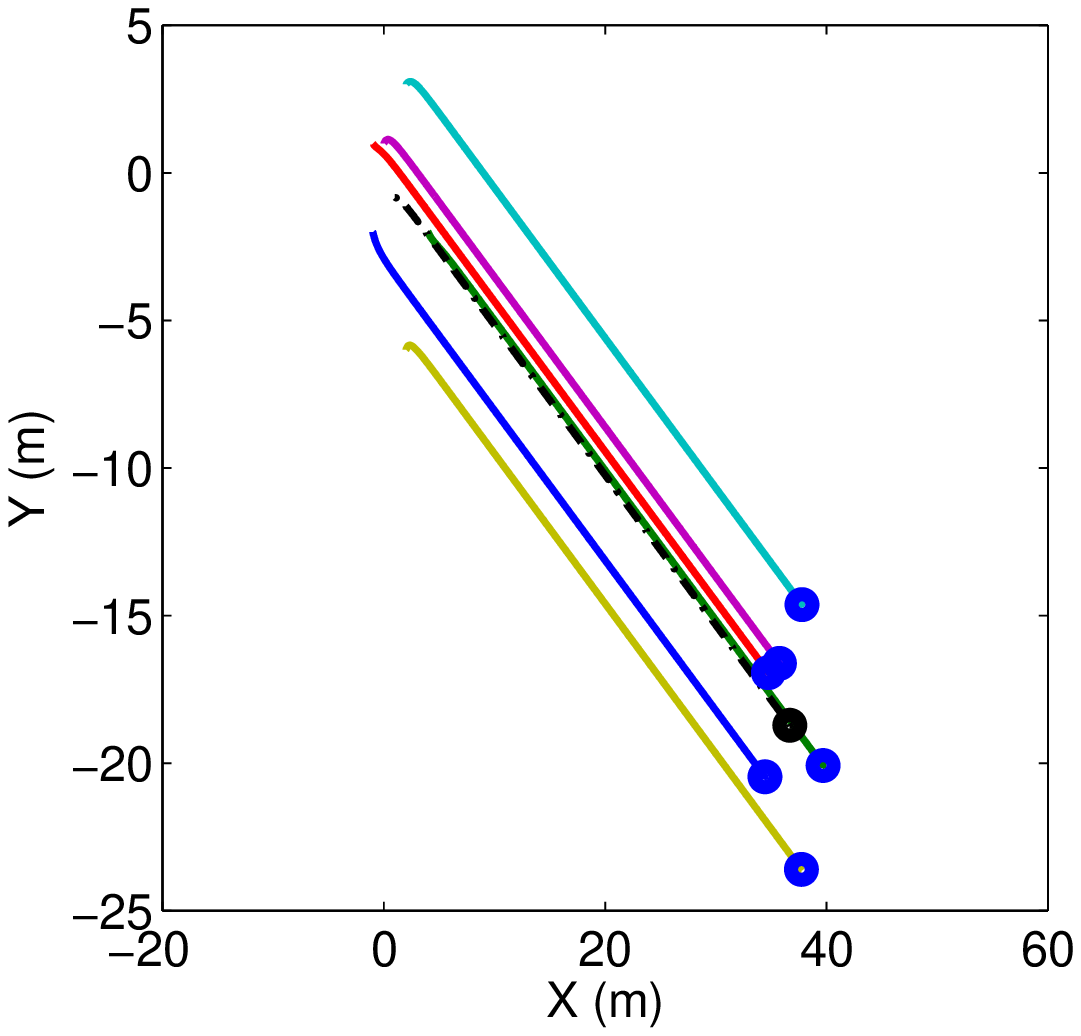}}\hspace{1cm}
\subfigure[]
{\includegraphics[scale=0.45]{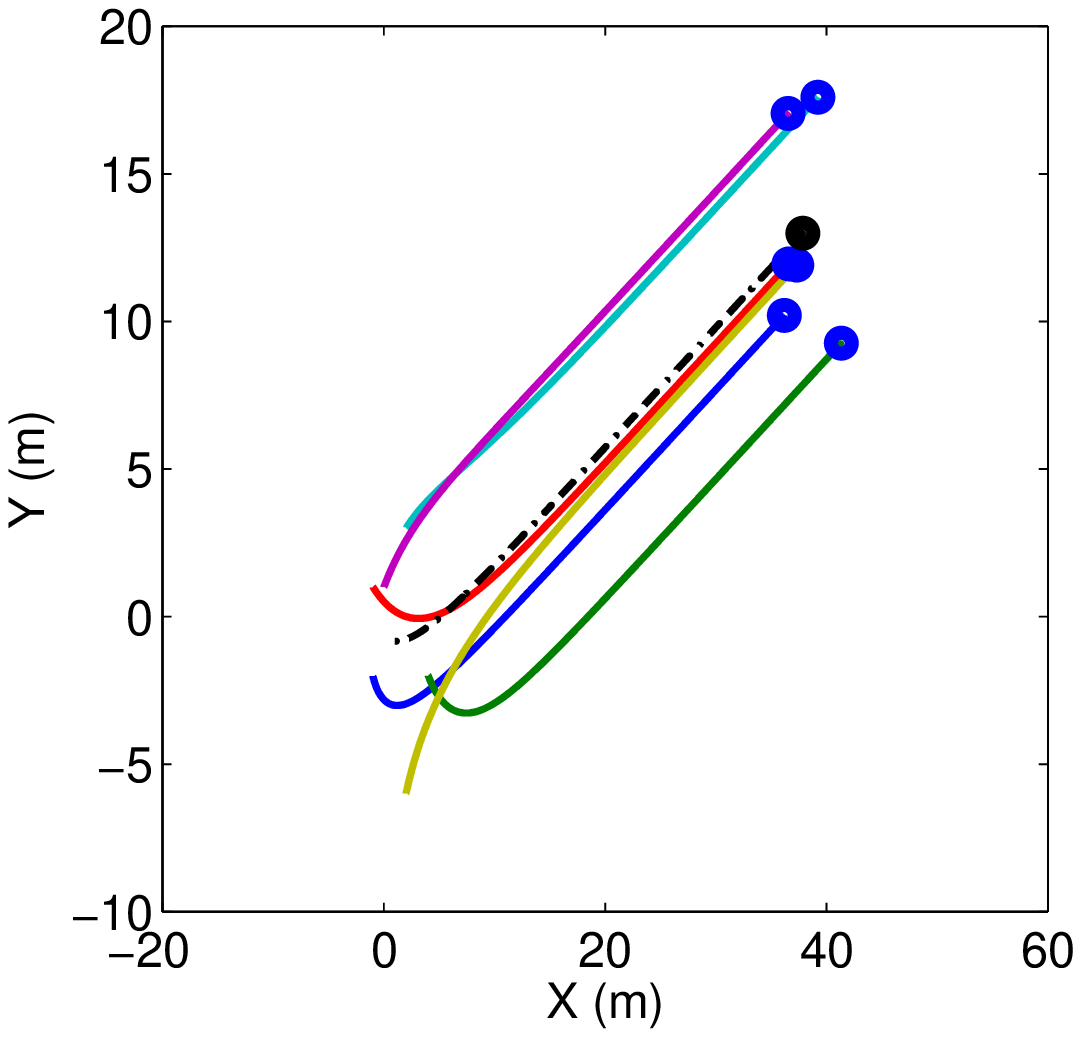}}\hspace{1cm}
\subfigure[]
{\includegraphics[scale=0.45]{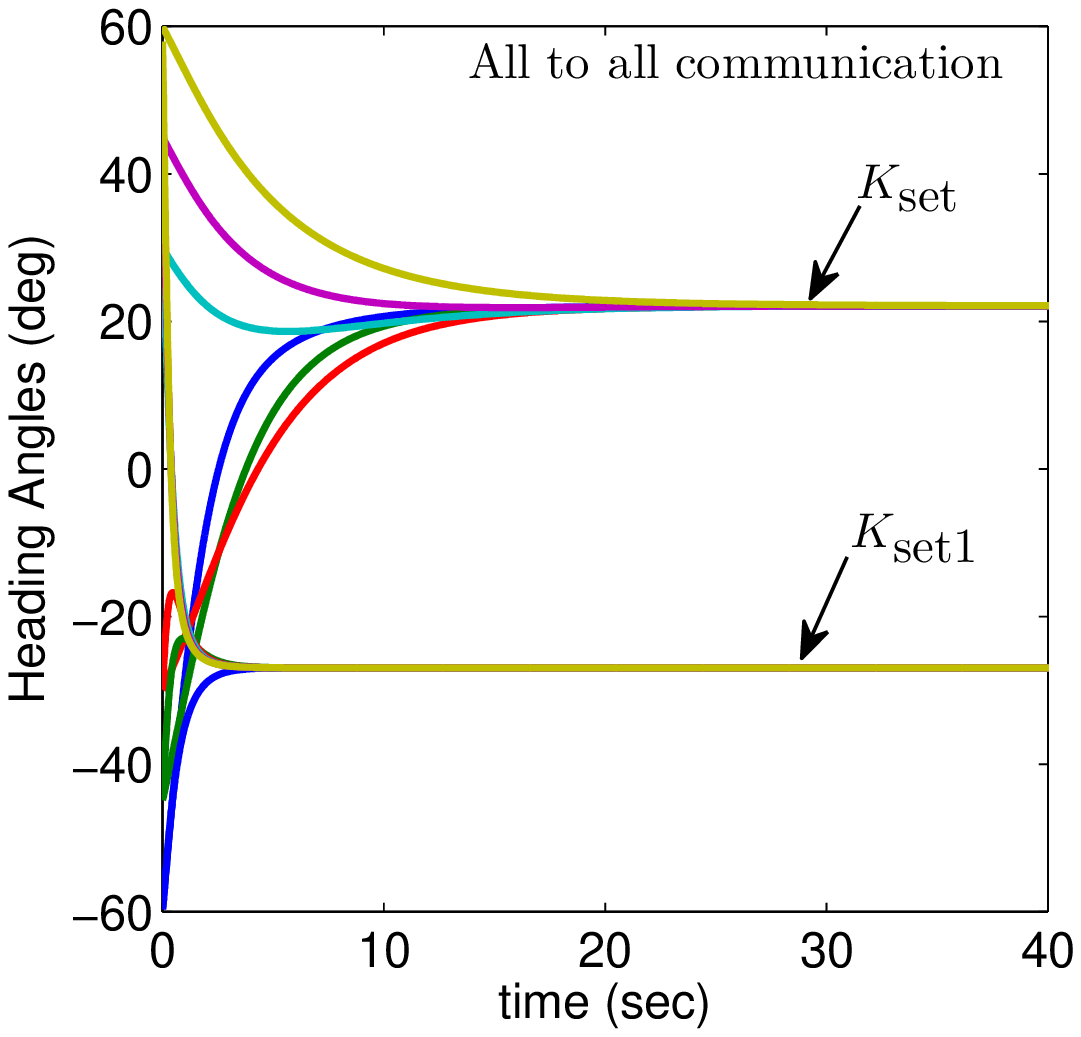}}\hspace{1cm}
\subfigure[]
{\includegraphics[scale=0.45]{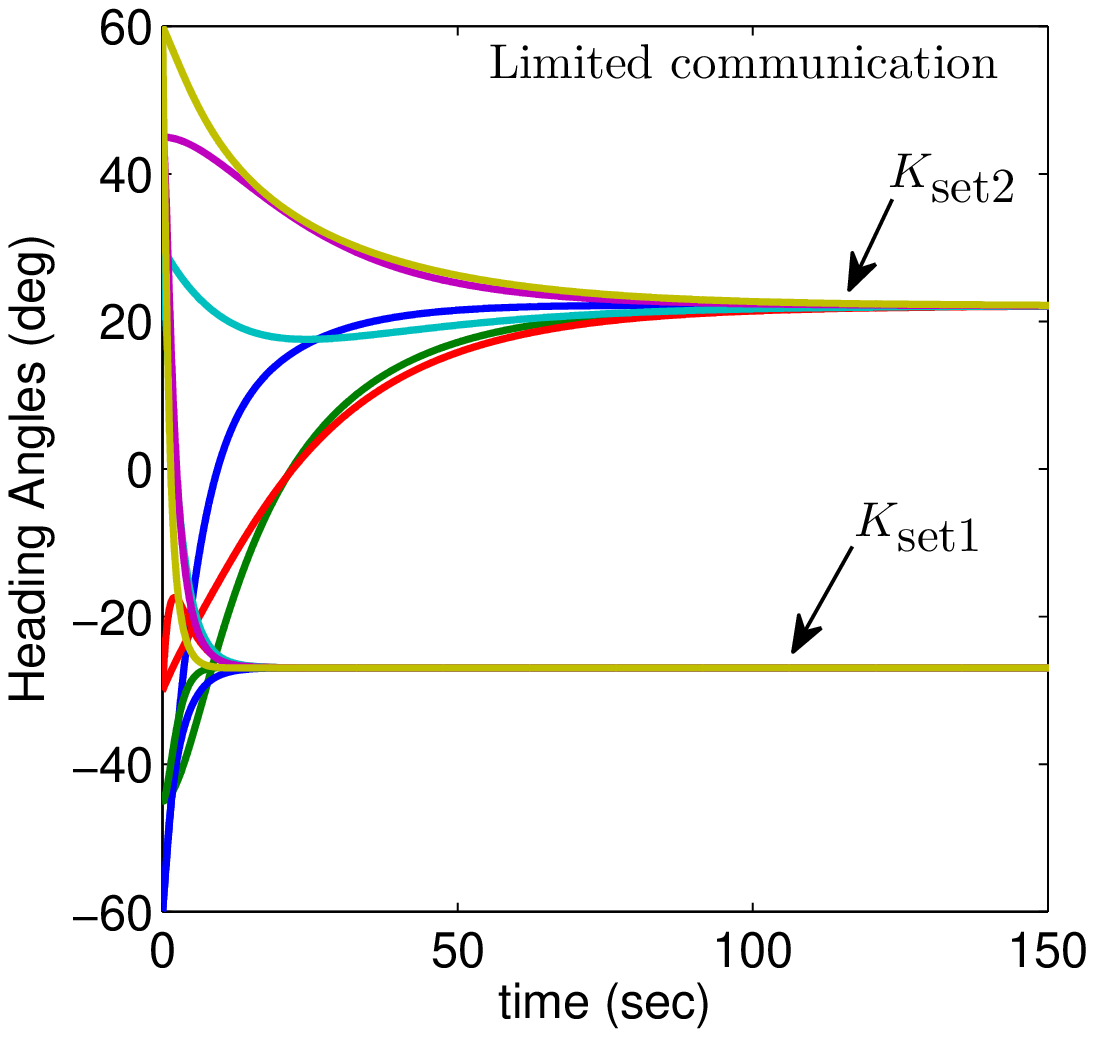}}
\caption{Synchronization of $N = 6$ agents for $\omega_0 = 0$ under the control laws \eqref{control3} and \eqref{control5}. $(a)$ Trajectories of the agents with $K_\text{set1} = \{K_k = -k,~k = 1, \ldots, 6\}$. $(b)$ Trajectories of the agents with $K_\text{set2} = \{K_k = -1/k,~k = 1, \ldots, 6\}$. $(c)$ Consensus of heading angles for the gains $K_\text{set1}$ and $K_\text{set2}$ under all-to-all communication. $(d)$ Consensus of heading angles for the gains $K_\text{set1}$ and $K_\text{set2}$ under limited communication.}
\label{Synchronization1}
\end{figure*}
\begin{figure*}
\centering
\subfigure[]{\includegraphics[scale=0.45]{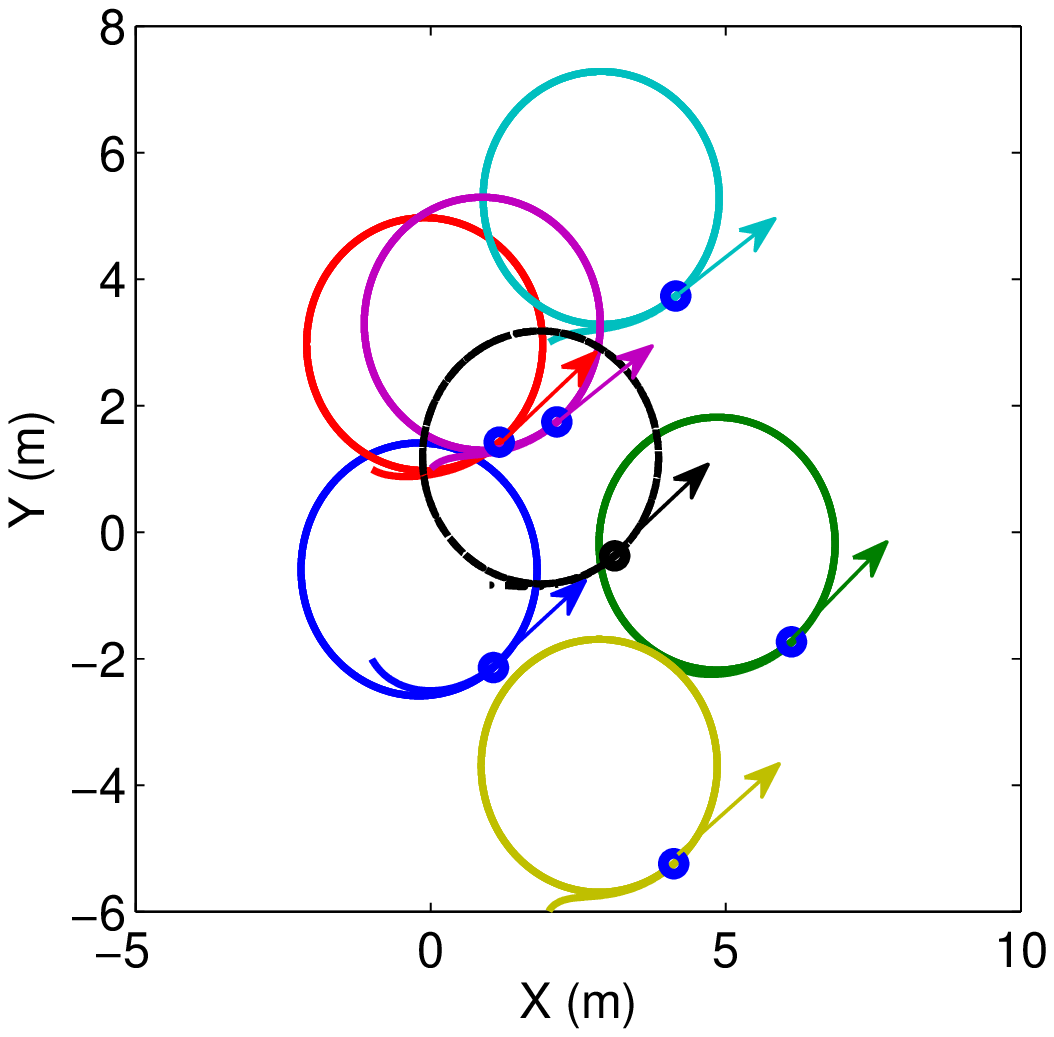}}\hspace{1cm}
\subfigure[]
{\includegraphics[scale=0.45]{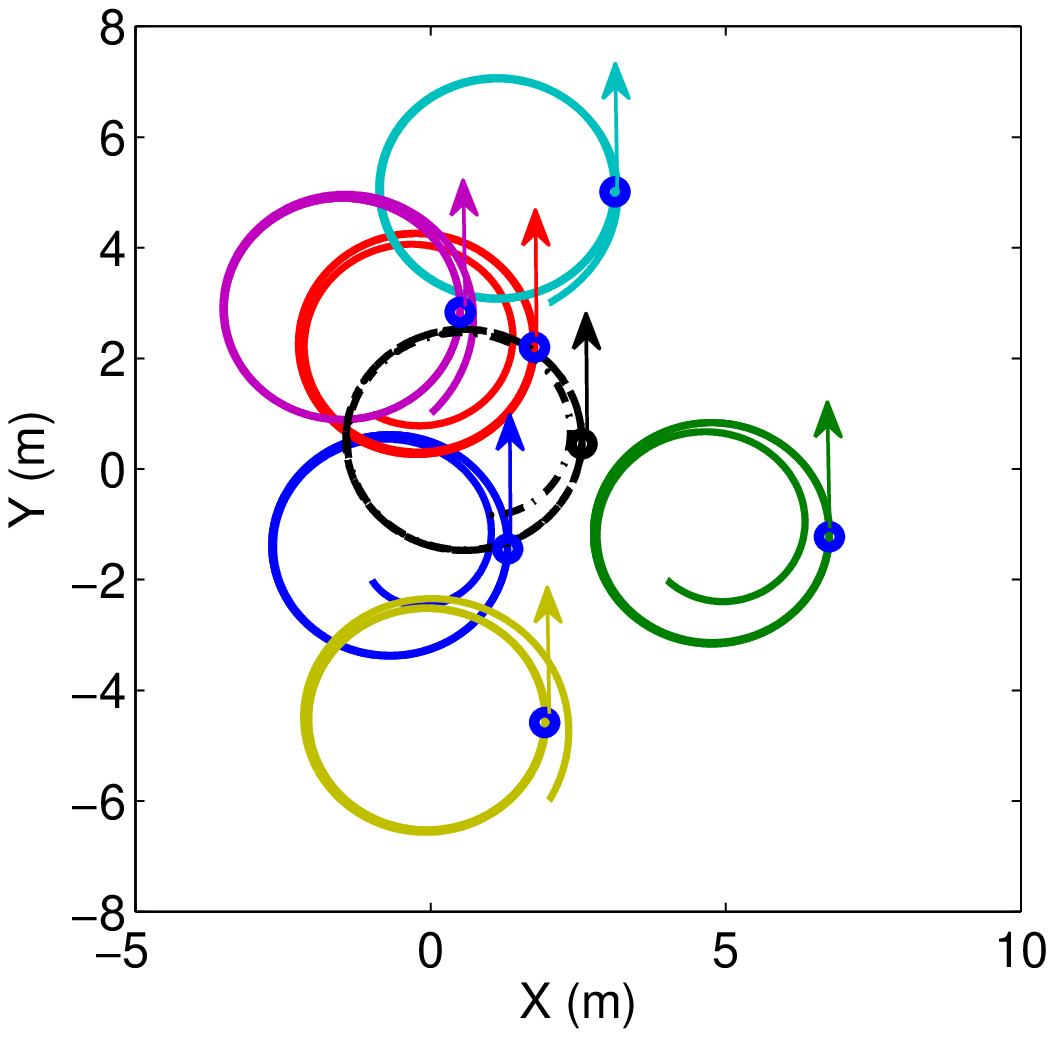}}\hspace{1cm}
\subfigure[]
{\includegraphics[scale=0.45]{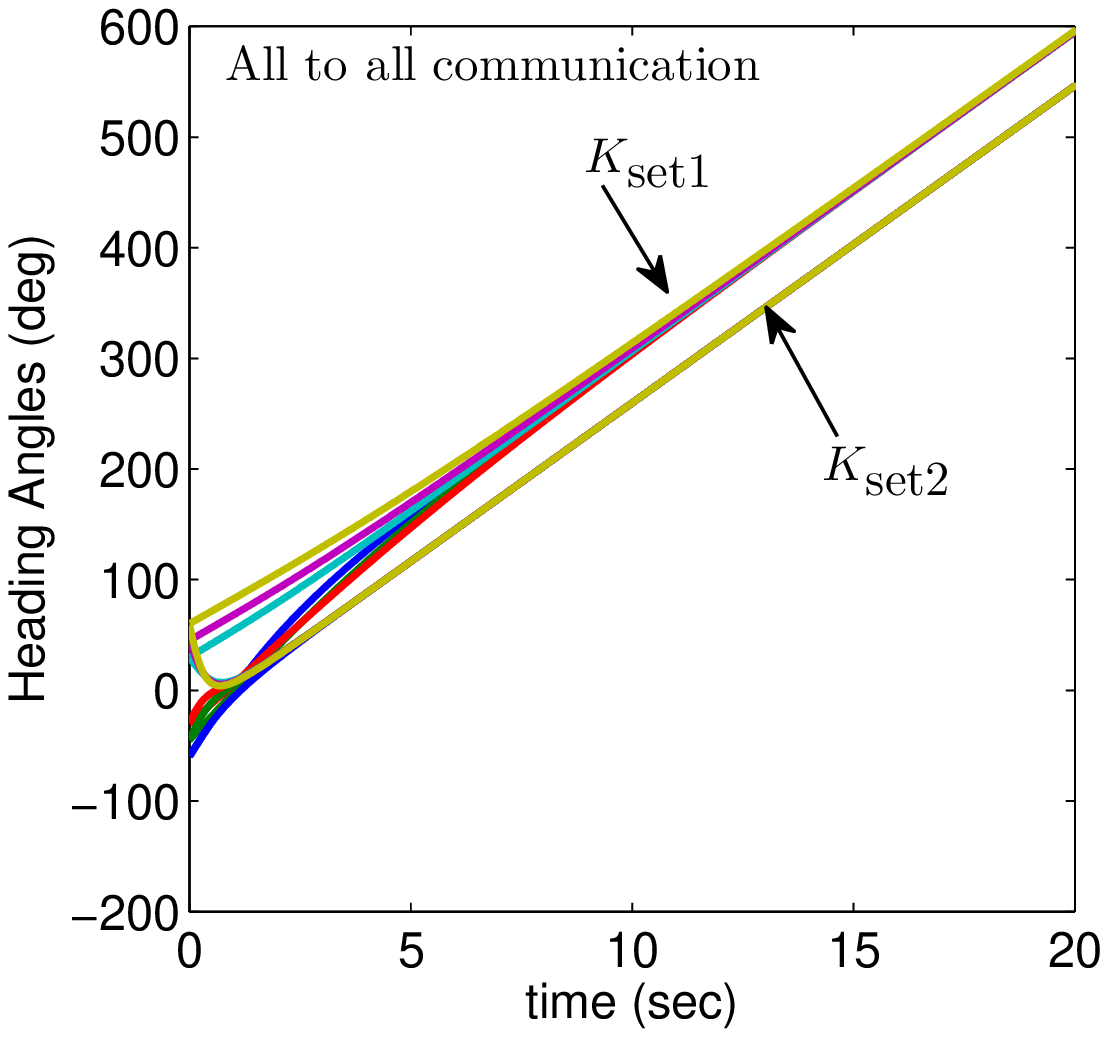}}\hspace{1cm}
\subfigure[]
{\includegraphics[scale=0.45]{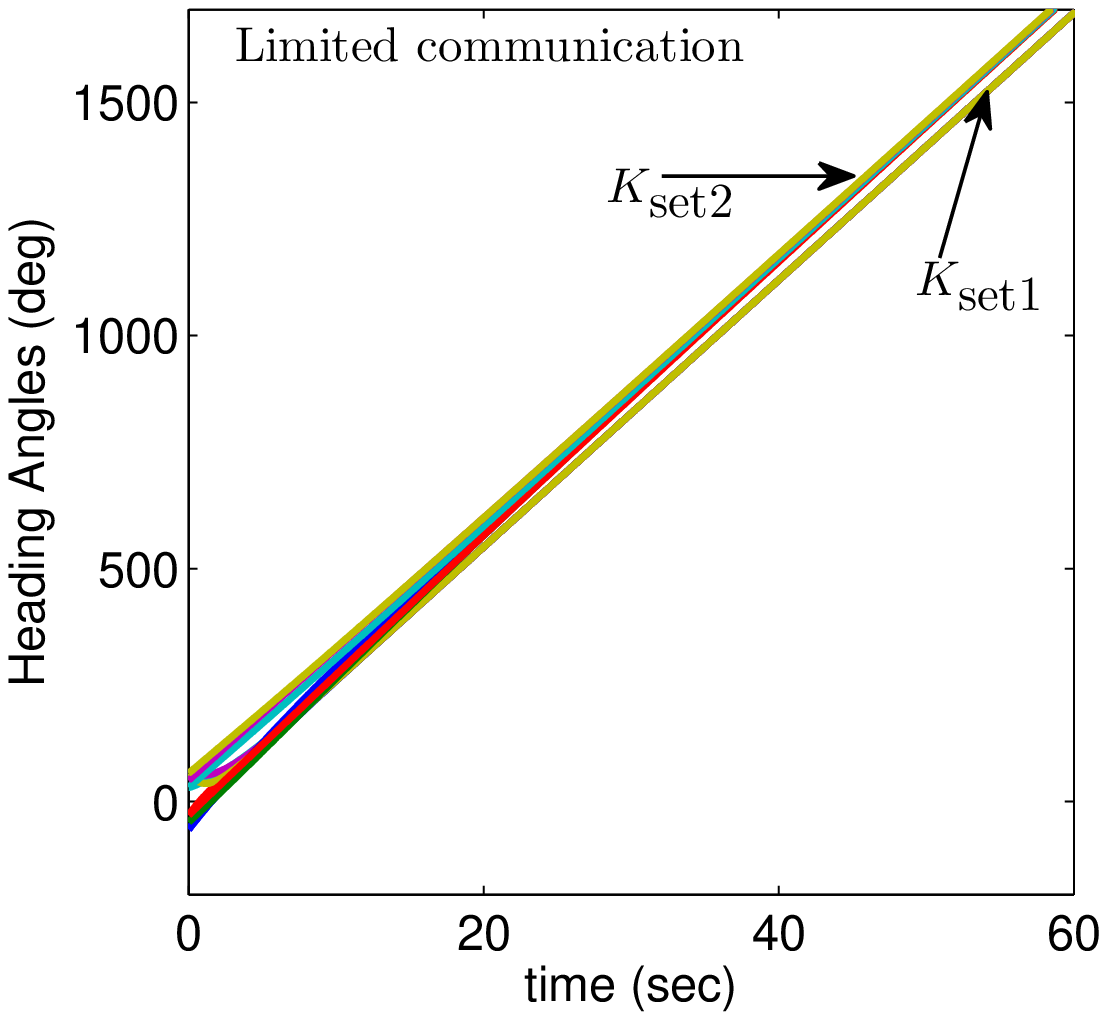}}
\caption{Synchronization of $N = 6$ agents for $\omega_0 = 0.5~\text{rad/sec}$ under the control laws \eqref{control2} and \eqref{control6}. $(a)$ Trajectories of the agents with $K_\text{set1}$. $(b)$ Trajectories of the agents with $K_\text{set2}$. $(c)$ Consensus of heading angles for the gains $K_\text{set1}$ and $K_\text{set2}$ under all-to-all communication. $(d)$ Consensus of heading angles for the gains $K_\text{set1}$ and $K_\text{set2}$ under limited communication.}
\label{Synchronization2}
\end{figure*}

\begin{proof}
Since the erroneous controller gain of the $k^\text{th}$ agent is $K \pm \epsilon_k$, by using \eqref{theta_c_new}, we can write
\begin{equation}
\label{error1}\hat{\theta}^p_c  = \left(\displaystyle\sum_{k=1}^{N} \frac{\hat{\theta}_{k0}}{K \pm \epsilon_k}\right) {\Big /} \left(\displaystyle\sum_{k=1}^{N} \frac{1}{K \pm \epsilon_k}\right).
\end{equation}
Since $\hat{\theta}_{k0}, \forall k$ are non-negative in the new coordinate system, the lower bound of $\hat{\theta}^p_c$, denoted by $\hat{\theta}^l_c$, is given by
\begin{equation}
\label{error2}\hat{\theta}^l_c  = \left(\displaystyle\sum_{k=1}^{N} \frac{\hat{\theta}_{k0}}{K + \epsilon_k}\right) {\Big /} \left(\displaystyle\sum_{k=1}^{N} \frac{1}{K - \epsilon_k}\right).
\end{equation}
Substituting $\epsilon_k = \eta_kK$ in \eqref{error2}, we get
\begin{eqnarray}
\hat{\theta}^l_c  &=& \left(\displaystyle\sum_{k=1}^{N} \frac{\hat{\theta}_{k0}}{1 + \eta_k}\right) {\Big /} \left(\displaystyle\sum_{k=1}^{N} \frac{1}{1 - \eta_k}\right)\\
& \geq & \left(\displaystyle\sum_{k=1}^{N} \frac{\hat{\theta}_{k0}}{1 + \eta}\right) {\Big /} \left(\frac{N}{1 - \eta}\right) =  \left(\frac{1 - \eta}{1 + \eta}\right)\overline{\hat{\theta}}_c.
\end{eqnarray}
Similarly, the upper bound of $\hat{\theta}^p_c$ is given by
\begin{eqnarray}
\hat{\theta}^u_c  &=& \left(\displaystyle\sum_{k=1}^{N} \frac{\hat{\theta}_{k0}}{1 - \eta_k}\right) {\Big /} \left(\displaystyle\sum_{k=1}^{N} \frac{1}{1 + \eta_k}\right)\\
& \geq & \left(\displaystyle\sum_{k=1}^{N} \frac{\hat{\theta}_{k0}}{1 - \eta}\right) {\Big /} \left(\frac{N}{1 + \eta}\right) = \left(\frac{1 + \eta}{1 - \eta}\right)\overline{\hat{\theta}}_c.
\end{eqnarray}
Thus, the maximum values of the lower and upper deviations of $\hat{\theta}^p_c$ from its mean value $\overline{\hat{\theta}}_c$ are, respectively,
\begin{eqnarray}
\Delta\hat{\theta}^l_c &=&  \overline{\hat{\theta}}_c - \left(\frac{1 - \eta}{1 + \eta}\right)\overline{\hat{\theta}}_c = \left(\frac{2\eta}{1 + \eta}\right)\overline{\hat{\theta}}_c\\
\Delta\hat{\theta}^u_c &=&  \left(\frac{1 + \eta}{1 - \eta}\right)\overline{\hat{\theta}}_c - \overline{\hat{\theta}}_c = \left(\frac{2\eta}{1 - \eta}\right)\overline{\hat{\theta}}_c.
\end{eqnarray}
It follows from the above discussion that
\begin{equation}
\hat{\theta}^p_c \in \left[\overline{\hat{\theta}}_c - \Delta\hat{\theta}^l_c, \overline{\hat{\theta}}_c + \Delta\hat{\theta}^u_c\right].
\end{equation}
However, since Theorem~\ref{Theorem4} ensures that $\hat{\theta}^p_c \in (\hat{\theta}_{m0}, \hat{\theta}_{M0})$ when there is heterogeneity in the controller gains, the actual set of angles reachable by $\hat{\theta}^p_c$ is \eqref{error_angle}. This completes the proof.
\end{proof}

\subsection{Case~2: $\omega_0 \neq 0$}
In this case, the motion of each agent is governed by \eqref{control2}. Thus, at equilibrium, the agents move in synchronization around their individual circular orbits at an angular frequency $\omega_0$. It implies that, in the steady state, the order parameter $p_\theta$ given by \eqref{phase_order_parameter} has constant, unit length and rotates at a constant angular frequency $\omega_0$. For ease of analysis in this framework, it is convenient to use a frame of reference that rotates at the same frequency $\omega_0$ so that $p_\theta$ remains stationary at the equilibrium where the system of agents synchronize. Thus, by replacing $\theta_k \rightarrow \theta_k + \omega_0 t$ in \eqref{control2}, which corresponds to a rotating frame at frequency $\omega_0$, we get the turn rate of the $k^\text{th}$ agent as
\begin{equation}
\dot{\theta}_k = -\frac{K_k}{N}\sum_{\substack{j=1, \\ j \neq k}}^{N}\sin(\theta_j - \theta_k),
\end{equation}
which is the same as \eqref{control3}. Therefore, all the analysis remains unchanged in a rotating frame of reference, and hence omitted.

{\it Simulation~1:} In this simulation, we consider $N=6$ agents with their initial positions and initial heading angles in the standard coordinate system, $\pmb{r}(0) = [(-1, -2), (4, -2), (-1, 1), (2, 3), (0, 1), (2, -6) ]^T$ and $\pmb{\theta}(0) = [-60^\circ, -45^\circ, -30^\circ, 30^\circ, 45^\circ, 60^\circ]^T$, respectively. Although the initial locations of the agents are given for representing the trajectories of the agents in the simulation, the locations themselves are not important so far as the objective of synchronization is concerned. Even with different locations, the convergence properties will be the same, although the trajectories will be different.

To account for the limited communication constraints, we present the simulations for a connected interaction network in which each agent is connected to its two neighbors only in a cyclic manner \cite{Marshall2004}.

In Fig.~\ref{Synchronization1}, synchronization of agents for the two sets of gains $K_\text{set1} = \{K_k = -k,~k = 1, \ldots, 6\}$, and $K_\text{set2} = \{K_k = -1/k,~k = 1, \ldots, 6\}$ is shown for $\omega_0 = 0$ under the controls \eqref{control3} and \eqref{control5}. The trajectories of the agents, in Figs.~$4(a)$ and $4(b)$, are shown only for the all-to-all communication scenario, and are similar for the limited communication case, and hence not shown. In all figures in this paper, the trajectory of the centroid is shown by a broken black line. The consensus in the heading angles of the agents for the two sets of gains is shown in Figs.~$4(c)$ and $4(d)$ for both types of communication scenarios, which indicates that different final velocity directions are achievable by using heterogeneous controller gains. Note that the convergence rate of the heading angles is faster under all-to-all interaction as expected.

Fig.~\ref{Synchronization2} depicts the synchronization of agents for the two sets of gains $K_\text{set1}$, and $K_\text{set2}$ for $\omega_0 = 0.5~\text{rad/sec}$ under the controls \eqref{control2} and \eqref{control6}. Here, all the agents, at any instant in time, are in synchronization, and move around individual circles of radius $\rho_0 = |\omega_0|^{-1} = 2~\text{m}$. The trajectories of the agents, in Figs.~$5(a)$ and $5(b)$, are again shown for the all-to-all communication scenario, and are similar for limited communication case.  In this case, since the agents continue to rotate around individual circles in a synchronized fashion, the final velocity direction keeps increasing with time, as shown in Figs.~$5(c)$ and $5(d)$ for both types of communication scenarios.

\section{A Special Case of Two Agents}
In this section, we address the special case of two agents and show that, unlike $K_k < 0, \forall k$, their exists a less restrictive condition on the heterogeneous gains $K_k$, which results in further expansion of the reachable set of the final velocity direction of the agents in their synchronized formation. We present the results only for $\omega_0 = 0$ since the analysis is unchanged for $\omega_0 \neq 0$ in a rotating frame of reference by redefining $\theta_k \rightarrow \theta_k + \omega_0 t$ for the $k^\text{th}$ agent.

For $N=2$, the time derivative of the potential function ${U}(\pmb{\theta})$ from \eqref{Udot_final} is given by
\begin{equation}
\label{U_dot_two_agents}\dot{U}(\pmb{\theta}) \big{|}_{N=2} = \frac{1}{2^2}(K_1 + K_2) \sin^2(\theta_2-\theta_1),
\end{equation}
which implies that the potential ${U}(\pmb{\theta})$ decreases if $K_1 + K_2 < 0$ since $\sin^2(\theta_2-\theta_1) > 0$. Moreover, it is easy to verify that $\sin^2(\theta_2-\theta_1) = 0$, only for the trivial cases when both the agents are already synchronized or balanced.

Thus, by using Theorem~\ref{Theorem1}, it follows from \eqref{U_dot_two_agents} that $K_1 + K_2 < 0$ is a sufficient condition to asymptotically stabilize the synchronized formation of $N=2$. Therefore, synchronized formation of $N=2$ is achievable for both positive and negative values of gains $K_1$ and $K_2$ provided that $K_1 + K_2 < 0$. Note that as there is only one communication link, both all-to-all and limited communication topologies are the same for $N=2$.


\begin{figure}[!t]
\centering
\includegraphics[scale=0.45]{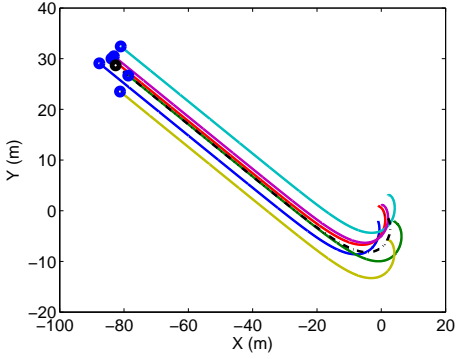}
\caption{Synchronization of $N=6$ agents under the control law \eqref{control3} with ${K}_\text{set3}$. Note that the final velocity direction lies outside the conic hull.}
\label{Synchronization3}
\end{figure}


\begin{remark}
For $N>2$, we did not come up with a simplified expression for the sufficient condition on the controller gains $K_k$, however, simulation results show that their exists a combination of both positive and negative values of the controller gains $K_k$ that gives rise to a synchronized formation with an extended set of reachable velocity direction $\hat{\theta}_c$. For example, the final velocity direction of the $6$ agents considered in Simulation~1 lies outside the conic hull for the set of gains
${K}_\text{set3} = \{K_1 = 0.5,~K_k = -k,~k=2, \ldots,6.\}$, and is shown in Fig.~\ref{Synchronization3}.
\end{remark}

Now, we state the following theorem, which says that the reachable set of the final velocity direction of the two agents in synchronization further expands when both positive and negative values of gains $K_1$ and $K_2$, satisfying $K_1 + K_2 < 0$, are selected.

\begin{thm}\label{Theorem6}
Consider 2 agents, with dynamics given by \eqref{modelNew}, under the control law \eqref{control3}. Then, any $\hat{\theta}_c \in [-\pi, \pi]$, which is the final velocity direction at which the system of agents synchronizes, is reachable iff there exist controller gains $K_1$ and $K_2$ such that $K_1 + K_2 < 0$.
\end{thm}

\begin{figure*}
\centerline{
\subfigure[]{\includegraphics[scale=0.38]{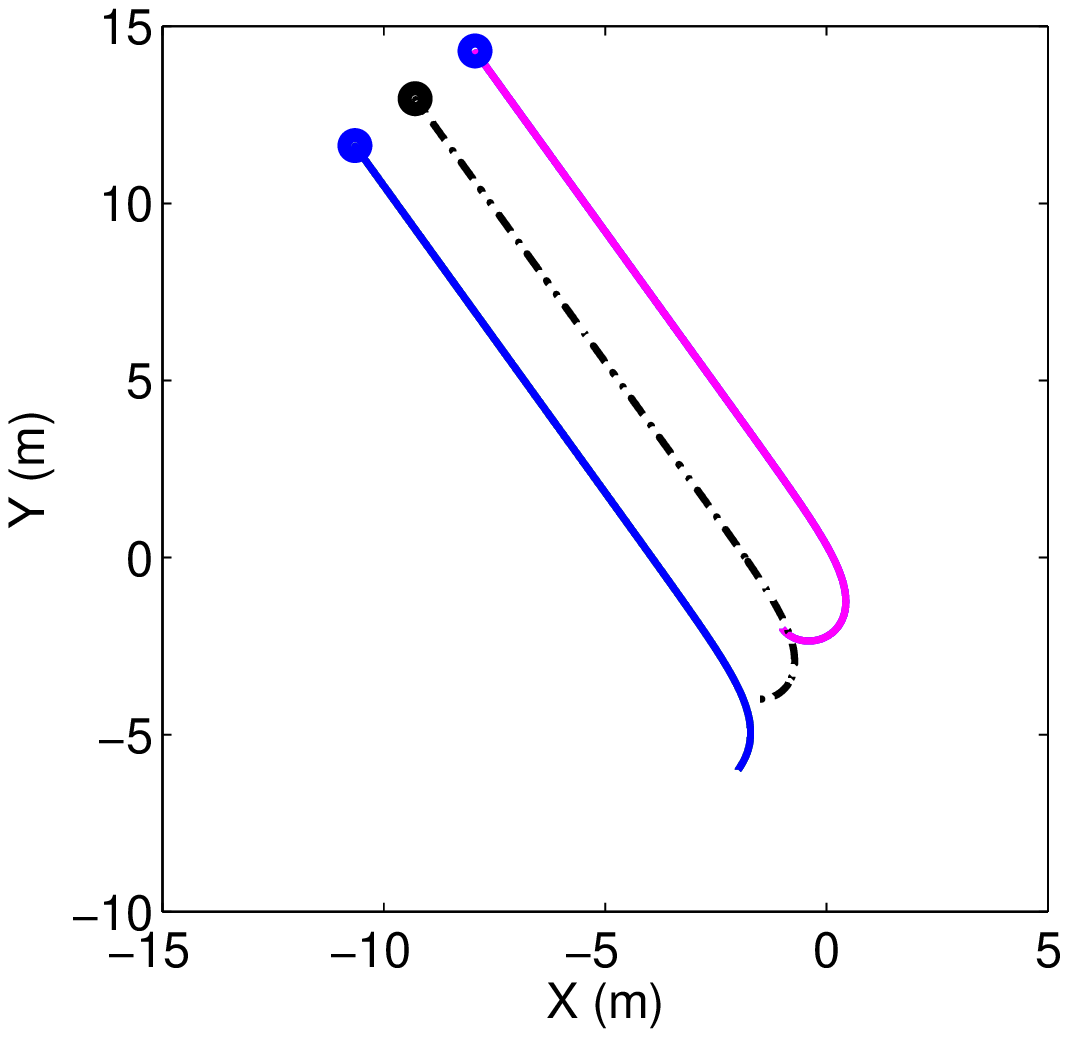}}\hspace{-0.5cm}
\subfigure[]
{\includegraphics[scale=0.38]{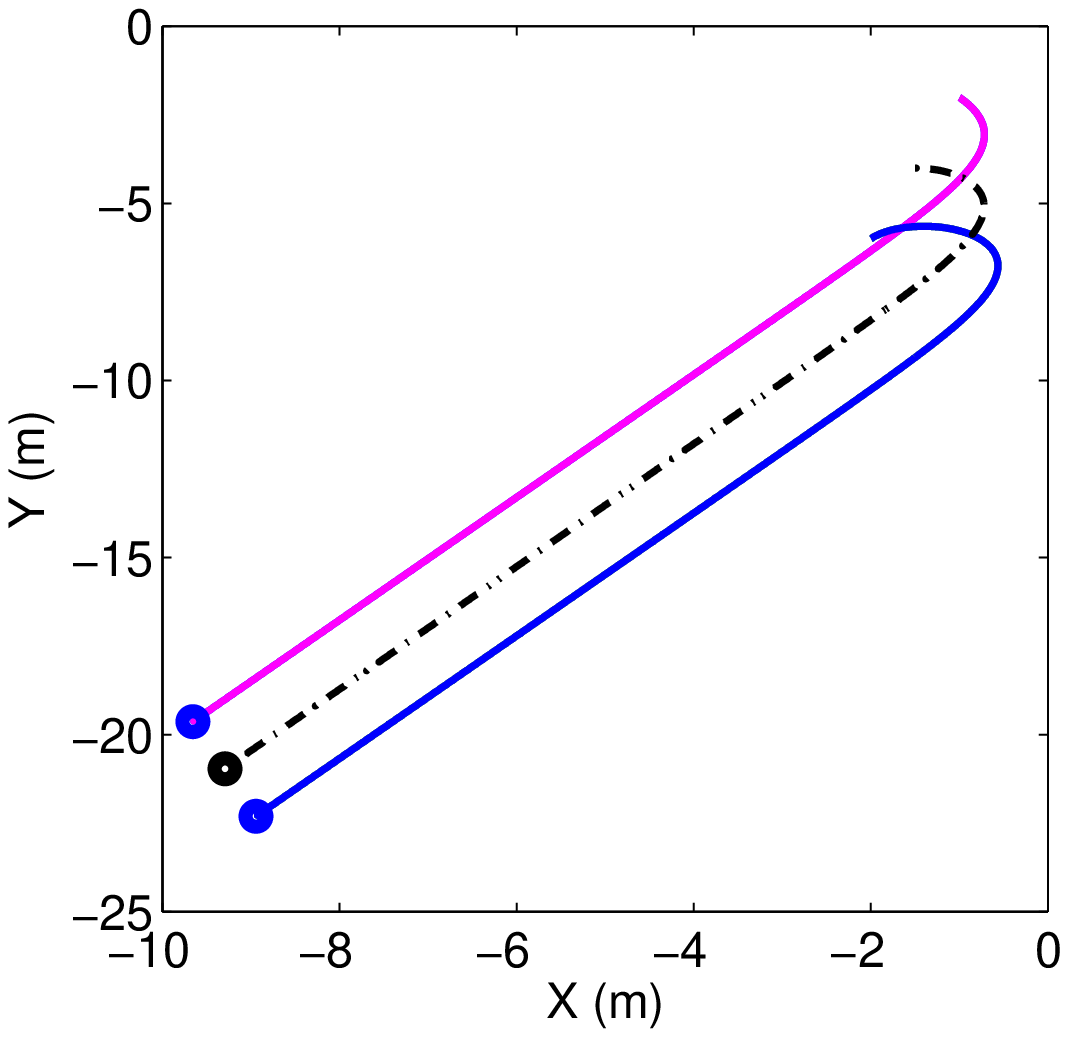}}\hspace{-0.5cm}
\subfigure[]
{\includegraphics[scale=0.38]{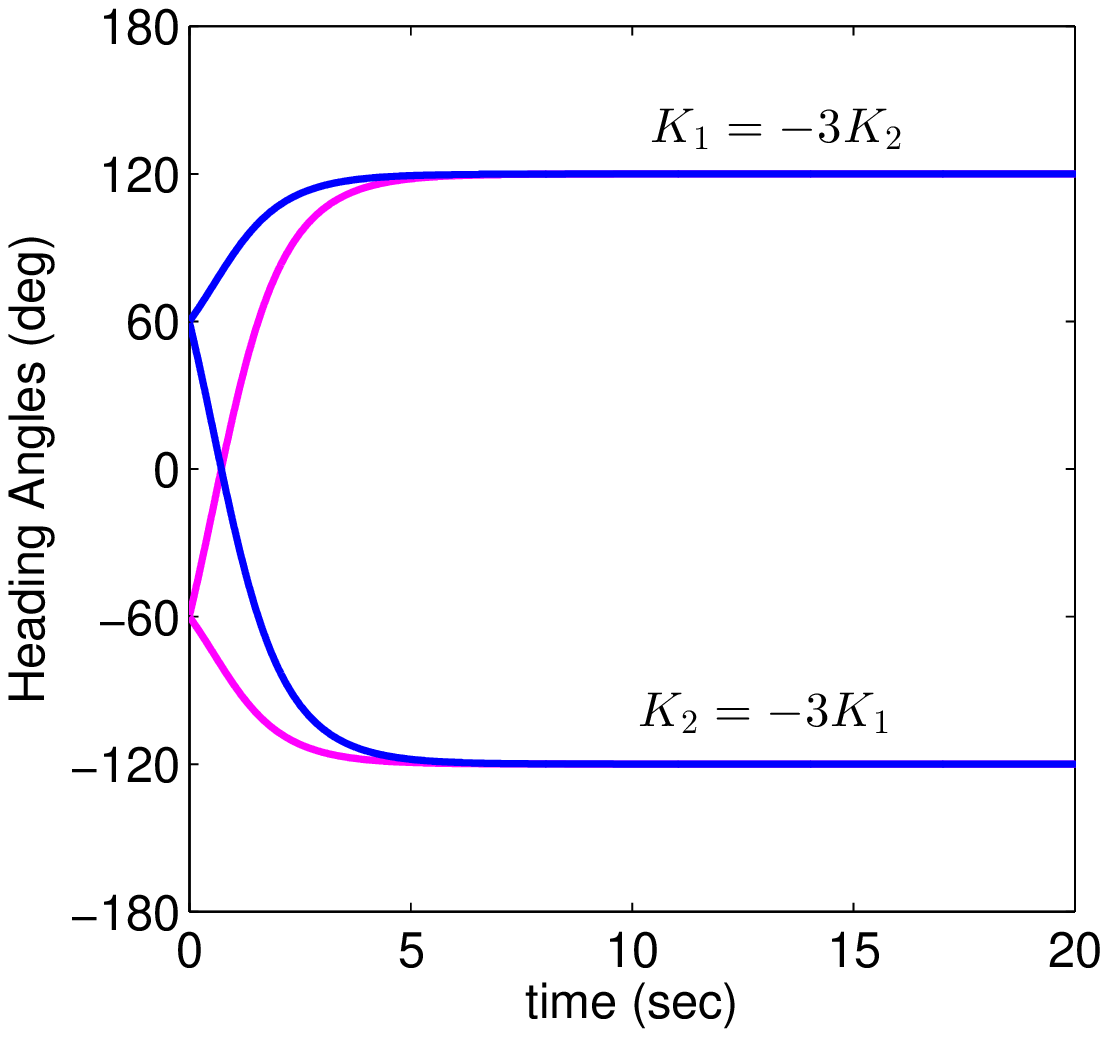}}}
\caption{Synchronization of $N = 2$ agents under the control law \eqref{control2}. $(a)$ Trajectories of the agents with $K_1 = -3K_2$. $(b)$ Trajectories of the agents with $K_2 = -3K_1$. $(c)$ Consensus of heading angles for the two sets of gains.}
\label{Synchronization_two_agents}
\end{figure*}

\begin{proof}
For $N=2$, the final velocity direction $\hat{\theta}_c$ of both the agents, by using \eqref{theta_c_new}, is given by
\begin{equation}
\label{theta_c_two_agents}\hat{\theta}_c =  \left(\dfrac{K_2}{K_1 + K_2}\right) \hat{\theta}_{10} + \left(\dfrac{K_1}{K_1 + K_2}\right) \hat{\theta}_{20}
\end{equation}

Substituting
\begin{equation}
\label{lamda1_2}\lambda_1 = \left(\dfrac{K_2}{K_1 + K_2}\right)~~ \text{and}~~\lambda_2 = \left(\dfrac{K_1}{K_1 + K_2}\right)
\end{equation}
in \eqref{theta_c_two_agents}, we get
\begin{equation}
\label{theta_c_two_agents_1}\hat{\theta}_c = \lambda_1\hat{\theta}_{10} + \lambda_2\hat{\theta}_{20}.
\end{equation}
Note that the parameters $\lambda_1$ and $\lambda_2$ satisfy $\lambda_1 + \lambda_2 = 1$. Without loss of generality, assume that $\hat{\theta}_{m0} = \hat{\theta}_{10}$ and $\hat{\theta}_{M0} = \hat{\theta}_{20}$. Now, depending upon the various choices of gains $K_1$ and $K_2$ satisfying $K_1 + K_2 < 0$, we consider the following three cases.

{\it Case~$1$:}
Let us assume that the gains $K_1 < 0$ and $K_2 < 0$. It implies that $\lambda_1 > 0$ and $\lambda_2 > 0$. In this situation, the proof directly follows from Theorem~\ref{Theorem4}, which ensures that $\hat{\theta}_c$ is reachable iff
\begin{equation}
\hat{\theta}_c \in (\hat{\theta}_{m0}, \hat{\theta}_{M0}).
\end{equation}

{\it Case~$2$:}
Assume that the gains $K_1 \geq 0$, $K_2 < 0$ and satisfy $K_1 + K_2 < 0$. It implies that $\lambda_1 > 0$ and $\lambda_2 \leq 0$. Thus, by using relation $\lambda_1 = 1 - \lambda_2$, \eqref{theta_c_two_agents_1} can be written as
\begin{equation}
\label{relation3_1}\hat{\theta}_c - \hat{\theta}_{10} = -\lambda_2(\hat{\theta}_{10} - \hat{\theta}_{20}).
\end{equation}
RHS (right-hand side) of \eqref{relation3_1} is non-positive, that is, $-\lambda_2(\hat{\theta}_{10} - \hat{\theta}_{20}) \leq 0$ since $\lambda_2 \leq 0$ and $\hat{\theta}_{10} < \hat{\theta}_{20}$ as per our assumption. Therefore, LHS (left-hand side) of \eqref{relation3_1} should also be non-positive, that is,
\begin{equation}
-\pi \leq \hat{\theta}_c \leq \hat{\theta}_{10}.
\end{equation}

{\it Case~$3$:}
Now, let us assume that the gains $K_1 < 0$, $K_2 \geq 0$ and satisfy $K_1 + K_2 < 0$. It implies that $\lambda_1 \leq 0$ and $\lambda_2 > 0$. Thus, by using relation $\lambda_2 = 1 - \lambda_1$, \eqref{theta_c_two_agents_1} can be written as
\begin{equation}
\label{relation4}\hat{\theta}_c - \hat{\theta}_{20} = \lambda_1(\hat{\theta}_{10} - \hat{\theta}_{20})
\end{equation}
RHS of \eqref{relation4} is non-negative, that is, $\lambda_1(\hat{\theta}_{10} - \hat{\theta}_{20}) \geq 0$ since $\lambda_1 \leq 0$ and $\hat{\theta}_{10} < \hat{\theta}_{20}$ as per our assumption. Therefore, LHS of \eqref{relation4} should also be non-negative, that is,
\begin{equation}
\hat{\theta}_{20} \leq \hat{\theta}_c \leq \pi.
\end{equation}

All the above cases lead to the conclusion that $\hat{\theta}_c \in [-\pi, \pi]$. This proves the necessary condition. To prove sufficiency condition for these two cases, we again consider the following cases.

{\it Case~$1$:}
Let $-\pi \leq \hat{\theta}_c \leq \hat{\theta}_{10}$ is reachable. Then according to \eqref{relation3_1}, the angular difference $\hat{\theta}_c - \hat{\theta}_{10}$ can be expressed as
\begin{equation}
\label{relation5}\hat{\theta}_c - \hat{\theta}_{10} = \beta(\hat{\theta}_{10} - \hat{\theta}_{20})
\end{equation}
where, $ \beta \geq 0$. Let us define $K_1 = \beta/c$ and $K_2 = -(1 + \beta)/c$, where $c > 0$ is a constant. Thus, $K_1 \geq 0$ and $K_2 < 0$ and satisfy $K_1 + K_2 = -({1}/{c})$.

Replacing $(1+\beta)$ and $\beta$ by $-cK_2$ and $cK_1$, respectively, in \eqref{relation5}, we get
\begin{equation}
\hat{\theta}_c = \left(\dfrac{K_2}{K_1 + K_2}\right)\hat{\theta}_{10} + \left(\dfrac{K_1}{K_1 + K_2}\right)\hat{\theta}_{20},
\end{equation}
which is the same as \eqref{theta_c_two_agents}.

{\it Case~$2$:} Let $ \hat{\theta}_{20} \leq \hat{\theta}_c \leq \pi$ is reachable. Then, according to \eqref{relation4}, the angular difference $\hat{\theta}_c - \theta_{20}$ can be expressed as
\begin{equation}
\label{relation6}\hat{\theta}_c - \hat{\theta}_{20} = -\gamma(\hat{\theta}_{10} - \hat{\theta}_{20})
\end{equation}
where, $ \gamma \geq 0$. Let us define $K_1 = -(1 + \gamma)/c$ and $K_2 = \gamma/c$, where $c > 0$ is a constant. Thus, $K_1 < 0$ and $K_2 \geq 0$ and again satisfy $K_1 + K_2 = -({1}/{c})$.

Replacing $\gamma$ and $(1+\gamma)$ by ${c}{K_2}$ and $-{c}{K_1}$, respectively in \eqref{relation6}, we again get \eqref{theta_c_two_agents}. These results imply that reachable set of final velocity directions further expands for $N=2$ when both positive and negative values of gains $K_1$ and $K_2$ satisfying $K_1 + K_2 < 0$ are selected. This completes the proof.
\end{proof}

{\it Simulation~2:}
In this simulation, we consider 2-agents with initial heading angles $\theta_{10} = -60^\circ$, and $\theta_{20} = 60^\circ$, and with randomly generated initial positions. By choosing a new coordinate system, one can easily compute from \eqref{theta_c_two_agents} that, if $K_1 = -3K_2$, the reachable velocity direction of both the agents in the standard coordinates is $\theta_c = 120^\circ$, and is shown in Figs.~$7(a)$ and $7(c)$, and if $K_2 = -3K_1$, $\theta_c = -120^\circ$, and is shown in Figs.~$7(b)$ and $7(c)$.

\section{Bounded Control Input}
In the previous sections, it has been assumed that the agents can use unbounded control input $u_k,\forall k$. However, in a practical scenario, autonomous vehicles, be it aerial, ground, or underwater, can develop only limited control force due to physical constraints. For example, one of the factors restricting the steering control for an unmanned aerial vehicle (UAV), is the bank angle of the aircraft. Since there is a finite limit to the degree to which a UAV can bank, the control force is bounded. To model this effect, a bound is placed on the turn rate $\dot{\theta}_k$, which can be done in the following two ways.

\subsection{Bounds on the controller gains $K_k$}
One of the ways to bound the control input is through the controller gains $K_k$. For example, one can observe from \eqref{control3} that
\begin{equation}
\label{u_bound}|\dot{\theta}_k| = |u_k| \leq \left(\frac{N-1}{N}\right)|K_k|
\end{equation}
since
\begin{equation}
-(N-1) \leq \sum_{\substack{j=1, \\ j \neq k}}^{N}\sin(\theta_j - \theta_k) \leq (N-1),
\end{equation}
for all $k = 1, \ldots, N$. Thus, the control input $u_k$ is bounded by the controller gain $K_k$. In order to bound the control input $u_k$ to a permissible limit, say
\begin{equation}
|u_k| \leq u_\text{max},
\end{equation}
where, $u_\text{max} > 0$ is the maximum allowable control for each agent, we can always choose the controller gain $K_k$ such that
\begin{equation}
\label{k_bound}|K_k| \leq \left(\frac{N}{N-1}\right)u_\text{max},
\end{equation}
$\forall~k$. From \eqref{u_bound} and \eqref{k_bound}, it follows that
\begin{equation}
|u_k| \leq \left(\frac{N-1}{N}\right)|K_k| \leq u_\text{max}.
\end{equation}
Thus, by bounding the controller gains $K_k$ according to \eqref{k_bound}, we can ensure that the control force $u_k$ does not violate the maximum allowable limit for the $k^\text{th}$ agent.

Note that, by selecting gains $K_k, \forall k$, appropriately, we can make the system of agents synchronize at a desired velocity direction with faster or slower convergence rates. Thus, by restricting the controller gains $K_k$ according to \eqref{k_bound}, convergence to synchronized formation may occur at a slower rate. To compensate for this, the control input may be bounded by the method described below.

\begin{figure*}
\centering
\subfigure[]{\includegraphics[scale=0.45]{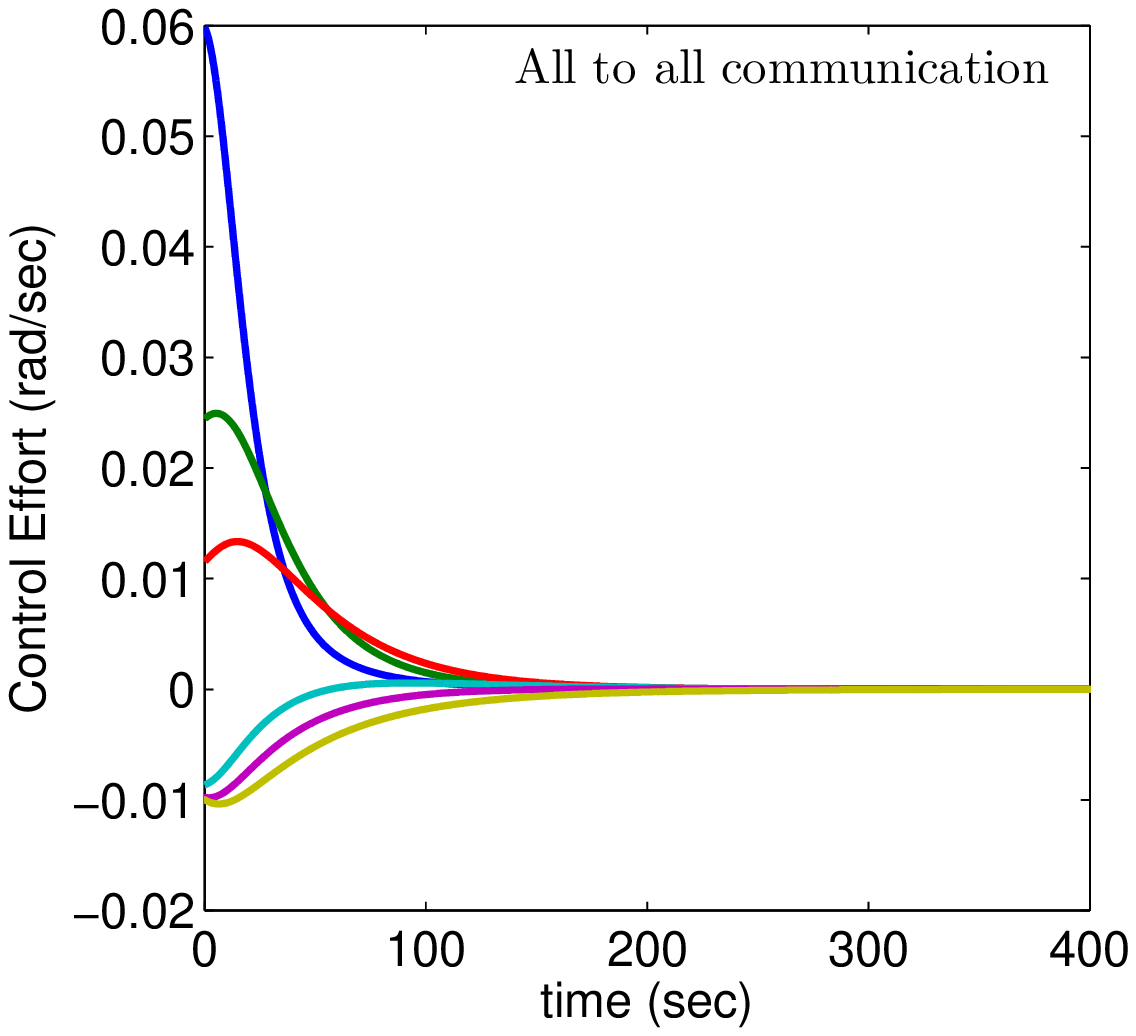}}\hspace{1cm}
\subfigure[]{\includegraphics[scale=0.45]{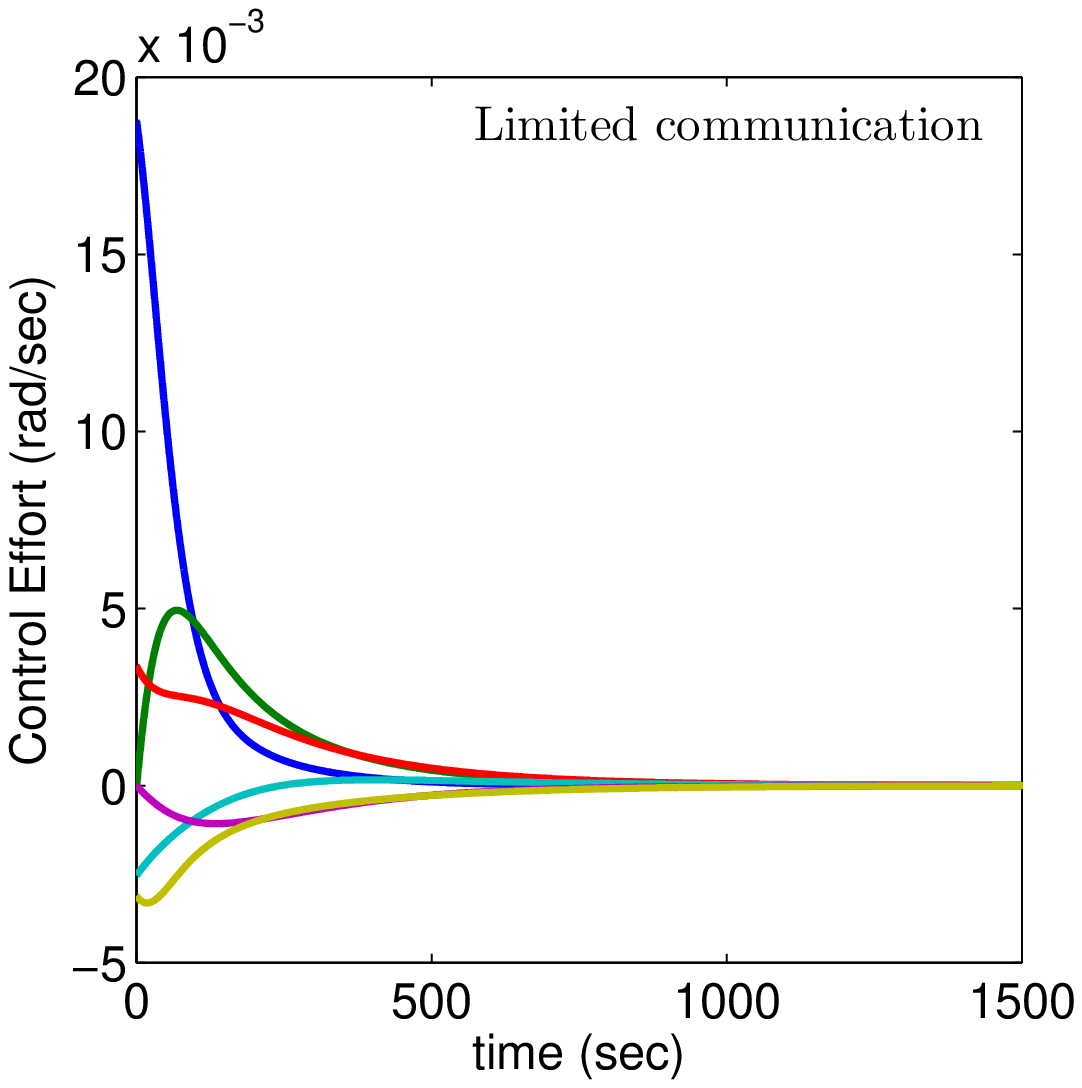}}\hspace{1cm}
\subfigure[]
{\includegraphics[scale=0.45]{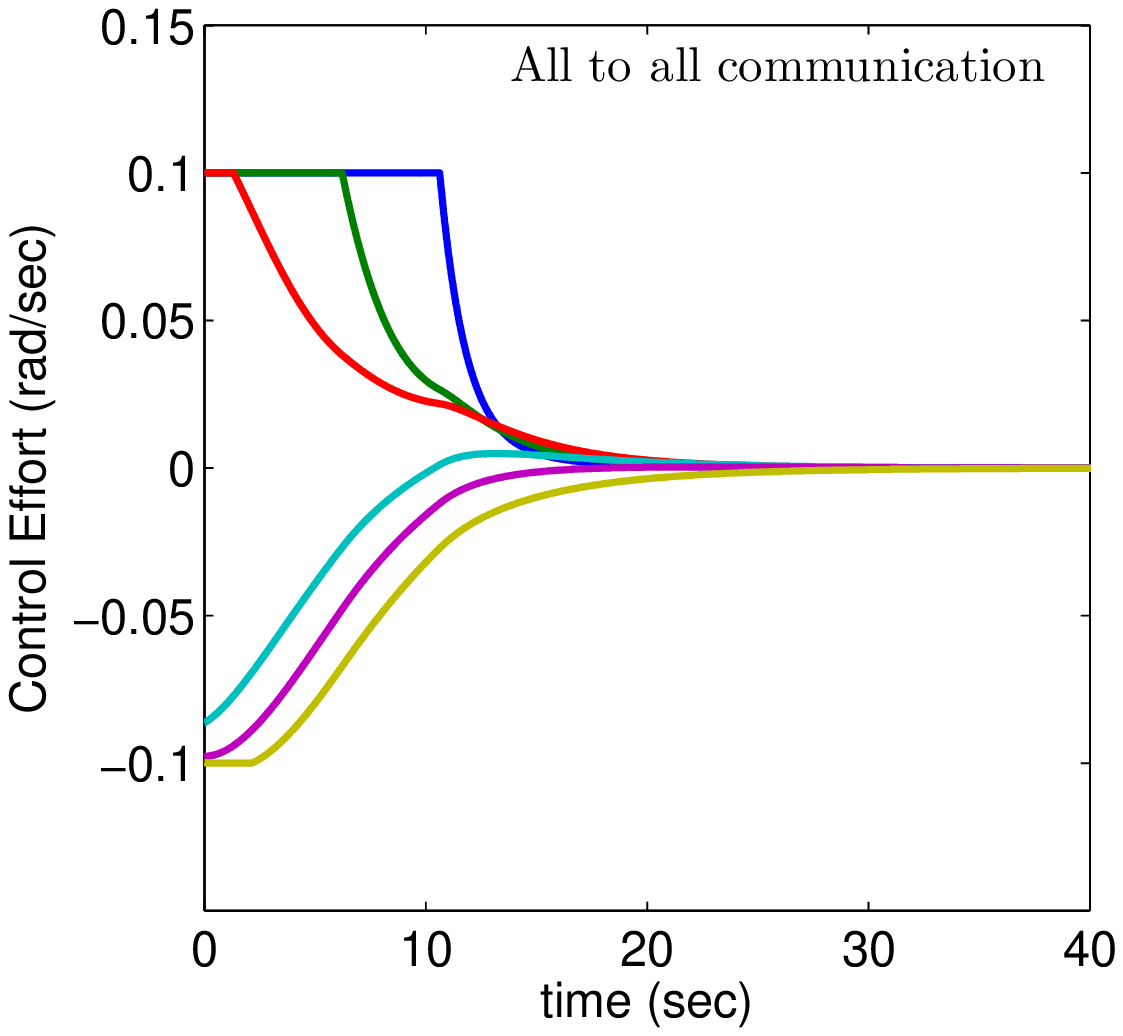}}\hspace{1cm}
\subfigure[]
{\includegraphics[scale=0.45]{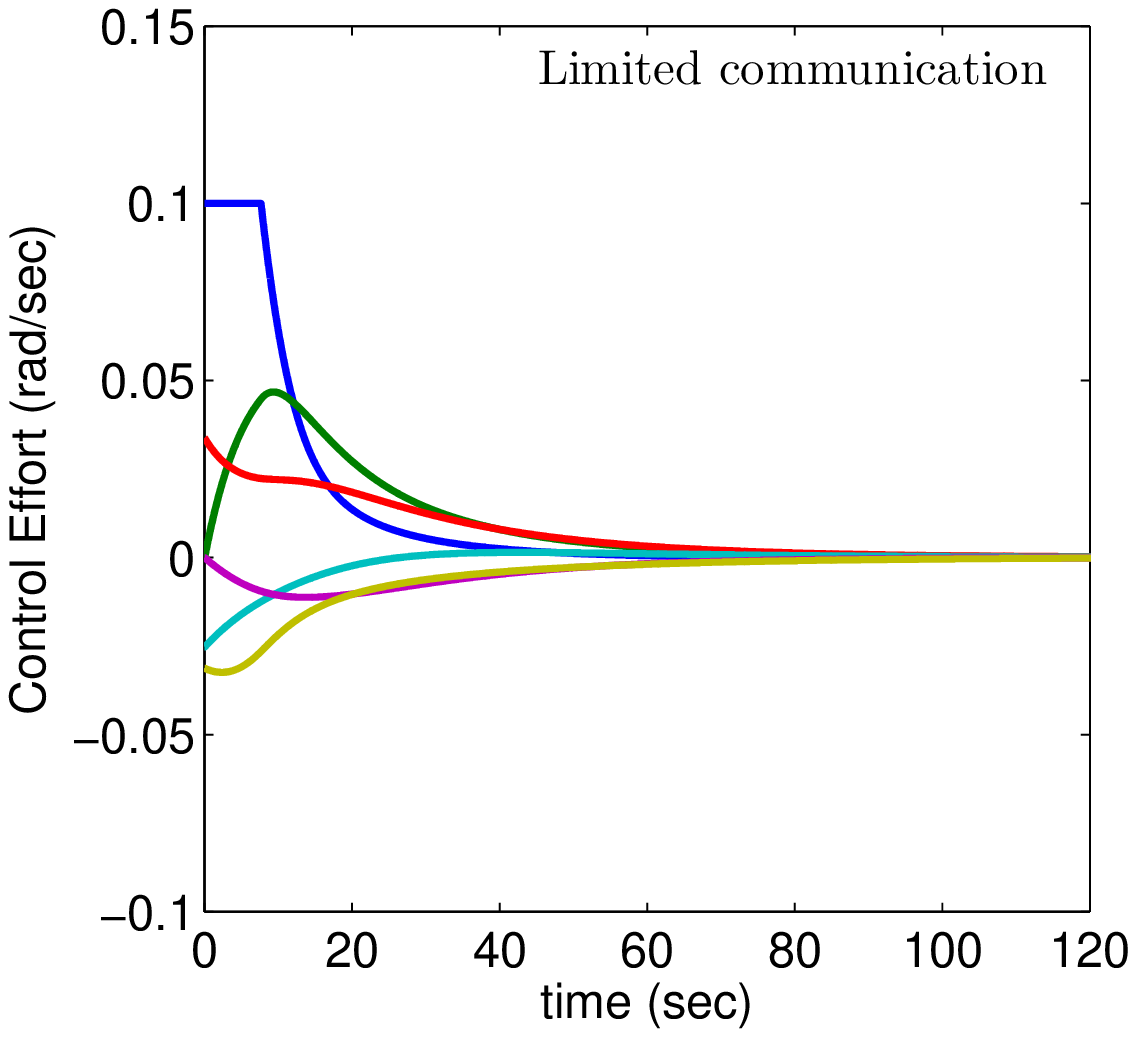}}
\caption{Control efforts of $N = 6$ agents in synchronized formation with a bound on the control effort given by $u_\text{max} = 0.1$. $(a)$ Control efforts given by \eqref{control3} with $K_\text{set4}$ under all-to-all communication. $(b)$ Control efforts given by \eqref{control5} with $K_\text{set4}$ under limited communication. $(c)$ Control efforts given by \eqref{u_saturation} with $K_\text{set2}$ under all-to-all communication. $(d)$ Control efforts given by \eqref{u_saturation} with $K_\text{set2}$ under limited communication.}
\label{bound on the control effort}
\end{figure*}

\subsection{Bounding $u_k$ by a saturation function}
Another method to bound the control input is by saturating $u_k$ for all $k$ according to the following saturation function \cite{Lee2001}:
\begin{equation}
\label{u_saturation}\dot{\theta}_k = \text{sat}(u_k;u_\text{max}) \triangleq
\begin{dcases}
    u_k , & \text{if}~|u_k| \leq u_\text{max}\\
    u_\text{max}~\text{sign}(u_k),  & \text{if}~|u_k| \geq u_\text{max}
\end{dcases}
\end{equation}
where, $\text{sign}(z)$ represents the signum function of $z$. Now, we state the following theorem which ensures stability of the synchronized formation under the control law \eqref{u_saturation}.

\begin{thm}\label{Theorem7}
Consider the system dynamics \eqref{modelNew} under the control law \eqref{u_saturation}, where, $u_k,~\forall k$, is given by \eqref{control1}. For $K_k < 0, \forall k$, all the agents asymptotically stabilize to a synchronized formation.
\end{thm}

\begin{proof}
Consider the potential function $U(\pmb{\theta})$ defined by \eqref{potential function}. Under the control law \eqref{u_saturation}, the time derivative of $U(\pmb{\theta})$
along the system dynamics \eqref{modelNew} is
\begin{equation}
\label{U_dot_saturation}\dot{U}(\pmb{\theta}) = \sum_{k=1}^{N} \left(\frac{\partial U}{\partial \theta_k}\right)\text{sat}(u_k;u_\text{max}).
\end{equation}
Using \eqref{control1}, \eqref{U_dot_saturation} becomes
\begin{equation}
\label{U_dot_saturation_1}\dot{U}(\pmb{\theta}) = \sum_{k=1}^{N} \frac{u_k\text{sat}(u_k;u_\text{max})}{K_k}.
\end{equation}
Substituting for $\text{sat}(u_k;u_\text{max})$ from \eqref{u_saturation} in \eqref{U_dot_saturation_1}, yields
\begin{equation}
\label{U_dot_saturation_2}\dot{U}(\pmb{\theta}) =
\begin{dcases}
    \sum_{k=1}^{N} \frac{u^2_k}{K_k}, & \text{if}~|u_k| \leq u_\text{max}\\
    u_\text{max}\sum_{k=1}^{N}\frac{u_k\text{sign}(u_k)}{K_k},  & \text{if}~|u_k| \geq u_\text{max}.
\end{dcases}
\end{equation}
Since $u_k\text{sign}(u_k) \geq 0, \forall k$, the condition $K_k < 0, \forall k$ ensures that $\dot{U}(\pmb{\theta}) \leq 0$. According to the LaSalle's invariance principle \cite{Khalil2000}, all the solutions of \eqref{modelNew} under control \eqref{u_saturation} converge to the largest invariant set contained in $\{\dot{U}(\pmb{\theta}) = 0\}$, which is the set of points where $u_k = 0,~\forall k$. Since $u_k = 0,~\forall k$ defines the critical set of $U(\pmb{\theta})$ (see \eqref{control1}), all the solutions of dynamics \eqref{modelNew} under control \eqref{u_saturation} asymptotically stabilize to the synchronized formation (Theorem~\ref{Theorem1}). This completes the proof.
\end{proof}

\begin{thm}\label{Theorem8}
Let $L$ be the Laplacian of an undirected an connected graph $\mathcal{G} = (\mathcal{V}, \mathcal{E})$ with $N$ vertices. Consider the system dynamics \eqref{modelNew} under the control law \eqref{u_saturation}, where, $u_k,~\forall k$, is given by \eqref{control4}. For $K_k < 0, \forall k$, all the agents asymptotically stabilize to a synchronized formation.
\end{thm}

\begin{proof}
Replacing $U(\pmb{\theta})$ by $W_L(\pmb{\theta})$ in the proof of Theorem~\ref{Theorem7} and then proceeding in the same way as above, we get the required result by using Theorem~\ref{Theorem2}.
\end{proof}

{\it Simulation~3:} In this simulation, we consider the same $6$ agents of Simulation~1. Let us assume that $u_\text{max} = 0.1$. At first, we obtain synchronization of all the agents as shown in Figs.~$8(a)$ and $8(b)$ under both types of communication scenarios for a set of gains $K_\text{set4} = \{K_k = -0.1/k,~k = 1, \ldots, 6.\}$, where all the gains $|K_k|, \forall k=1, \ldots, 6$, are bounded below by $u_\text{max} = 0.1$. On the other hand, by saturating control efforts according to \eqref{u_saturation} for all $k = 1, \ldots, 6$, synchronization of all the agents for the set of gains $K_\text{set2}$ is shown in Fig.~$8(c)$ and $8(d)$ under both types of communication scenarios. Note that the synchronization is faster for $K_\text{set2}$ as well as for all to all interaction, as desired. Nevertheless the convergence for the control \eqref{u_saturation} is faster, we cannot assure synchronization at a desired velocity direction, since \eqref{u_saturation} does not results in \eqref{theta_c_new}. Thus, in this situation, a desired final velocity direction can be obtained only by bounding the heterogeneous controller gains according to \eqref{u_bound} but at a slower convergence rate compared to \eqref{u_saturation}.

\section{Conclusions}
In this paper, we have investigated the phenomenon of synchronization for a group of heterogeneously coupled agents. It has been shown that a desired final velocity direction of the agents in synchronization can be achieved by appropriately selecting the heterogeneous controller gains $K_k$ satisfying $K_k < 0, \forall k$. Moreover, it has been illustrated through simulation that the reachable set of the final velocity direction further expands when both positive and negative values of the heterogeneous gains are incorporated in the control scheme. In particular, it has been proved analytically for $N=2$ that there exists a condition on the heterogeneous controller gains which allows them to assume both positive and negative values, and hence results in further expansion of the reachable set of the final velocity direction. We have further discussed the synchronization of realistic systems where an upper bound on the control force, applied to each agent, was obtained either by bounding the heterogeneous controller gains or by directly saturating the control efforts. In a nutshell, the idea of introducing heterogeneity in the controller gains works well in choosing a desired final velocity direction in the synchronized formation of the agents, and hence is more useful in various practical applications. A possible future work is to explore the behavior of the system under heterogeneous controller gains with time-varying interaction among agents.

\section{Acknowledgements}
This work is partially supported by Asian Office of Aerospace Research and Development (AOARD).



\end{document}